\documentclass{article}

\usepackage{cmbright}
\usepackage{amssymb,amsmath,amsthm}
\usepackage{mathrsfs}
\usepackage{color}
\usepackage[colorlinks,allcolors=blue]{hyperref}
\usepackage[T1]{fontenc}

\def\B{\mathscr B}
\def\C{\mathbb C}
\def\d{\mathrm{d}}
\def\E{\mathbb E}
\def\F{\mathscr F}
\def\G{\mathcal G}
\def\H{\mathcal H}
\def\K{\mathcal K}
\def\KK{\mathscr K}
\def\N{\mathbb N}
\def\NN{\mathfrak N}
\def\P{\mathbb P}

\def\R{\mathbb R}

\def\T{\mathbb T}
\def\U{\mathrm U}
\def\V{\mathtt V}
\def\X{\mathtt X}
\def\Z{\mathbb Z}

\def\dom{\mathcal D}
\def\supp{\mathop{\mathrm{supp}}\nolimits}
\def\Ran{\mathop{\mathrm{Ran}}\nolimits}
\def\e{\mathop{\mathrm{e}}\nolimits}

\def\lone{\mathop{\mathrm{L}^1}\nolimits}
\def\ltwo{\mathop{\mathrm{L}^2}\nolimits}

\DeclareMathOperator*{\slim}{s\hspace{0.1pt}-\hspace{0.1pt}lim}

\newtheorem{Theorem}{Theorem}[section]
\newtheorem{Remark}[Theorem]{Remark}
\newtheorem{Lemma}[Theorem]{Lemma}
\newtheorem{Corollary}[Theorem]{Corollary}
\newtheorem{Proposition}[Theorem]{Proposition}
\newtheorem{Definition}[Theorem]{Definition}
\newtheorem{Assumption}[Theorem]{Assumption}
\newtheorem{Example}[Theorem]{Example}

\begin{document}


\title{Quantum walks with an anisotropic coin II\;\!: scattering theory}

\author{S. Richard$^1$\footnote{Supported by the grant
\emph{Topological invariants through scattering theory and noncommutative geometry} from Nagoya University,
and on leave of absence from Univ.~Lyon, Universit\'e Claude Bernard Lyon
1, CNRS UMR 5208, Institut Camille Jordan, 43 blvd.~du 11 novembre 1918, F-69622
Villeurbanne cedex, France.},
A. Suzuki$^2$\footnote{Supported by JSPS Grant-in-Aid for Young Scientists B no 26800054.},
R. Tiedra de Aldecoa$^3$\footnote{Supported by the Chilean Fondecyt Grant 1170008.}}

\date{\small}
\maketitle \vspace{-1cm}

\begin{quote}
\emph{
\begin{itemize}
\item[$^1$] Graduate school of mathematics, Nagoya University,
Chikusa-ku,\\
Nagoya 464-8602, Japan
\item[$^2$] Division of Mathematics and Physics, Faculty of Engineering, Shinshu
University, Wakasato, Nagano 380-8553, Japan
\item[$^3$] Facultad de Matem\'aticas, Pontificia Universidad Cat\'olica de Chile,\\
Av. Vicu\~na Mackenna 4860, Santiago, Chile
\item[] \emph{E-mails:} richard@math.nagoya-u.ac.jp, akito@shinshu-u.ac.jp,
rtiedra@mat.puc.cl
\end{itemize}
}
\end{quote}


\begin{abstract}
We perform the scattering analysis of the evolution operator of quantum walks with an
anisotropic coin, and we prove a weak limit theorem for their asymptotic velocity. The
quantum walks that we consider include one-defect models, two-phase quantum walks, and
topological phase quantum walks as special cases. Our analysis is based on an abstract
framework for the scattering theory of unitary operators in a two-Hilbert spaces
setting, which is of independent interest.
\end{abstract}

\textbf{2010 Mathematics Subject Classification:} 46N50, 47A40, 47B47, 60F05.

\smallskip

\textbf{Keywords:} Quantum walks, unitary operators, scattering theory,  weak limit
theorem


\section{Introduction}\label{sec_intro}
\setcounter{equation}{0}

This paper is motivated by recent works on topological phenomena for quantum walks
\cite{AB13,CGSVWW,GNVW,Ki06,Kit12,KBFRBKADW12,KRBD10,Ob15}. It is the second part of a
series of papers on one-dimensional quantum walks with an anisotropic behaviour at
infinity. In our first paper \cite{RST_1}, we performed the spectral analysis of the
quantum walks and we developed abstract commutator methods for unitary operators in a
two-Hilbert spaces setting. Here we pursue our study by investigating the scattering
theory of the quantum walks and establishing a weak limit theorem \cite{Kon02, Kon05}.
We also present a suitable abstract framework for the proof of the existence and
completeness of wave operators for unitary operators in a two-Hilbert spaces setting.

The one-dimensional anisotropic quantum walks that we consider are described by a
unitary operator $U:=SC$ in the Hilbert space $\H:=\ell^2(\Z,\C^2)$, where $S$ is a
shift operator and $C$ is a coin operator acting by multiplication by unitary matrices
$C(x)\in\U(2)$, $x\in\Z$, with short-range asymptotics at infinity:
\begin{equation}\label{eq_SR}
C(x) =
\begin{cases}
C_\ell+O(|x|^{-1-\varepsilon_\ell}) & \hbox{as $x\to-\infty$}\\
C_{\rm r}+O(|x|^{-1-\varepsilon_{\rm r}}) & \hbox{as $x\to+\infty$}
\end{cases}
\quad\hbox{with}\quad
C_\ell,C_{\rm r}\in\U(2),~\varepsilon_\ell,\varepsilon_{\rm r}>0.
\end{equation}
The assumption \eqref{eq_SR} covers a wide range of quantum walks such as homogeneous
(or translation-invariant) quantum walks \cite{ABNVW,GJS,Kon02,Kon05}, one-defect
models \cite{CGMV12,Kon10,KLS13,WLKGGB12}, and two-phase quantum walks
\cite{EEKST15,EEKST16,EKH15}. Some classes of inhomogeneous (or position-dependent)
quantum walks \cite{ABJ,Suz16} and split-step quantum walks \cite{Kit12} are also
covered by our assumption. We refer to the introduction of \cite{RST_1} for additional
references on earlier works.

A weak limit theorem for quantum walks is a result of the following type: If $\X_n$
denotes the random variable for the position of a quantum walker at time $n\in\Z$,
then $\X_n/n$ converges in law to a random variable $\V$ as $n\to\infty$. Since
$\X_n/n$ is the average velocity of the quantum walker, the random variable $\V$ can
be interpreted as the asymptotic velocity of the walker. It is therefore of particular
interest to determine the density function of the probability distribution $\mu_\V$ of
$\V$. As was put into evidence in \cite{Suz16}, where the second author considered the
case $C_\ell=C_{\rm r}$ in \eqref{eq_SR}, the key ingredients for the proof of the
weak limit theorem are the following:
\begin{enumerate}
\item[(i)] absence of singular continuous spectrum for $U$,
\item[(ii)] existence of an asymptotic velocity operator $V_{\rm ac}$ for $U$.
\end{enumerate}
Once these assertions are proved, a weak limit theorem can be established in a way
similar to \cite{Suz16} and the distribution $\mu_\V$ can be expressed as
\begin{equation}\label{eq_mu_V}
\mu_\V=\big\|E_{\rm p}^U\Psi_{\rm in}\big\|_\H^2\;\!\delta_0
+\big\|E^{V_{\rm ac}}(\;\!\cdot\;\!)E_{\rm ac}^U\Psi_{\rm in}\big\|_\H^2,
\end{equation}
with $\delta_0$ the Dirac measure for the point $0$, $\Psi_{\rm in}\in\H$ the initial
state of the walker, $E_{\rm p}^U$ and $E_{\rm ac}^U $ the projections onto the pure
point and absolutely continuous subspaces of $U$, and $E^{V_{\rm ac}}$ the spectral
measure of $V_{\rm ac}$. In our first paper \cite{RST_1}, we proved the assertion (i)
and provided information on the eigenvalues of $U$ by constructing a conjugate
operator $A$ for $U$ under the assumption \eqref{eq_SR}. Here, we build on the results
of \cite{RST_1} to prove the assertion (ii) and to establish a detailed formula for
the distribution \eqref{eq_mu_V}.

The organisation of the paper is the following. In Section \ref{sec_scattering}, we
develop our framework for the scattering theory for unitary operators in a two-Hilbert
spaces setting. Given two unitary operators $U$ and $U_0$ acting in Hilbert spaces $\H$
and $\H_0$, and a bounded operator $J:\H_0\to\H$, we establish in Theorem
\ref{thm_wave} and Corollary \ref{corol_wave} criteria for the existence and
completeness under smooth perturbations of the local wave operators
\begin{equation}\label{eq_wo}
W_\pm(U,U_0,J,\Theta):=\slim_{n\to\pm\infty}U^{-n}JU_0^n E^{U_0}(\Theta),
\end{equation}
where $E^{U_0}$ is the spectral measure of $U_0$ and
$\Theta\subset\T:=\{z\in\C\mid|z|=1\}$ an open set. These results for the scattering
theory of unitary operators in two Hilbert spaces are new in such a generality. They
are a natural analogue of similar results for the scattering theory of self-adjoint
operators in two Hilbert spaces, which can be found for example in
\cite{RT13_2,Yaf92}.

In Section \ref{sec_walks}, we apply our results on scattering theory to anisotropic
quantum walks with full evolution operator $U=SC$ and free evolution operator
$U_0:=U_\ell\oplus U_{\rm r}$, where $U_\ell:=SC_\ell$ and $U_{\rm r}:=SC_{\rm r}$
describe the behaviour of the quantum walker as $x\to-\infty$ and $x\to+\infty$. We
prove in Theorem \ref{thm_complete_walks} the existence and completeness of the wave
operators for the pair $\{U_0,U\}$, and in Proposition \ref{prop_kernels} we give a
description of the initial sets of the wave operators in terms of the velocity
operators $V_\ell,V_{\rm r}$ for the evolution operators $U_\ell,U_{\rm r}$.

Section \ref{sec_weak_lim} is dedicated to the proof of the weak limit theorem for the
anisotropic quantum walks. First, we prove in Proposition \ref{prop_Heisenberg} the
assertion (ii) above, that is, the existence of an asymptotic velocity operator
$V_{\rm ac}$ for the full evolution operator $U$. We show that $V_{\rm ac}$ is given
by
$$
V_{\rm ac}=W_+(U,U_0,J,\Theta)V_0\;\!W_+(U,U_0,J,\Theta)^*,
$$
with $V_0:=V_\ell\oplus V_{\rm r}$ the asymptotic velocity operator for the free
evolution operator $U_0$. Then, in Theorem \ref{thm_wlt} we use the results (i) and
(ii) above to prove the weak limit theorem, and in Theorem \ref{Thm_wlt_meas} we
establish an explicit formula for the density function of the probability distribution
$\mu_\V$ given in \eqref{eq_mu_V}. Namely, if we set $a_\ell:=|(C_\ell)_{1,1}|$ and
$a_{\rm r}:=|(C_{\rm r})_{1,1}|$, let $f_{\rm K}:\R\times(0,1]\to[0,\infty)$ be the
Konno function, and write $\chi_B$ for the characteristic function for
a set $B$, then we prove that $\V$ has distribution
\begin{align}
\mu_\V(\d\upsilon)
&=\kappa_0\;\!\delta_0(\d\upsilon)+\kappa_\ell\;\!\delta_{-1}(\d\upsilon)
+\kappa_{\rm r}\;\!\delta_1(\d\upsilon)\nonumber\\
&\quad+\chi_{[-a_\ell,0)}(\upsilon)w_\ell(\upsilon)\tfrac12f_{\rm K}(\upsilon,a_\ell)
\;\!\d\upsilon+\chi_{(0,a_{\rm r}]}(\upsilon)w_{\rm r}(\upsilon)\tfrac12
f_{\rm K}(\upsilon,a_{\rm r})\;\!\d\upsilon,\label{eq_dmu_V}
\end{align}
with $\kappa_0:=\|E^U_{\rm p}\Psi_{\rm in}\|_\H^2$, $\kappa_\ell,\kappa_{\rm r}\ge0$
and $w_\ell,w_{\rm r}:\R\to[0,\infty)$. In addition, we show that $\kappa_\ell$ is
nontrivial when $a_\ell=1$, $\kappa_{\rm r}$ is nontrivial when $a_{\rm r}=1$,
$w_\ell$ is nontrivial and has support in $[-a_\ell,0)$ when $a_\ell\in(0,1)$, and
$w_{\rm r}$ is nontrivial and has support in $(0,a_{\rm r}]$ when $a_{\rm r}\in(0,1)$.
See Theorem \ref{Thm_wlt_meas} for the explicit formulas of
$\kappa_\ell,\kappa_{\rm r}$ and $w_\ell,w_{\rm r}$. We also show that the
decomposition \eqref{eq_dmu_V} of $\mu_\V$ is unique.

An interpretation of the formula \eqref{eq_dmu_V} detailed after Theorem \ref{thm_wlt}
and in Example \ref{ex_wlt} is the following. Localisation occurs if the probability
that the asymptotic velocity vanishes is positive, i.e., $\P(\V=0)>0$. Since
\eqref{eq_dmu_V} implies that $\P(\V=0)=\kappa_0=\|E^U_{\rm p}\Psi_{\rm in}\|_\H^2$,
localisation occurs if and only if the initial state $\Psi_{\rm in}$ has an overlap
with the pure point subspace of $U$. Furthermore, the quantum walker moves
asymptotically to the left at speed $\upsilon\in[-a_\ell,0)$ if $a_\ell\in(0,1)$ and
at speed $\upsilon=1$ if $a_\ell=1$. Similarly, the quantum walker moves
asymptotically to the right at speed $\upsilon\in(0,a_{\rm r}]$ if $a_{\rm r}\in(0,1)$
and at speed $\upsilon=1$ if $a_{\rm r}=1$. In particular, if $a_\ell>a_{\rm r}$, then
the quantum walker can move faster on the left-hand side than on the right-hand side.

Finally, in Example \ref{ex_cpr} at the end of Section \ref{sec_weak_lim}, we explain
how our formula \eqref{eq_dmu_V} for the distribution $\mu_\V$ generalises several
formulas already available in the literature. For example, it generalises a similar
formula for isotropic quantum walks where $C(x)=C_\infty+O(|x|^{-1-\varepsilon})$
\cite{Suz16}, which include one-defect models \cite{KLS13} and homogeneous quantum
walks \cite{Kon02,Kon05,GJS}. The formula \eqref{eq_dmu_V} also generalises the
formula obtained in \cite{EEKST16} for two-phase quantum walks where $C(x)=C_-$ for
$x\le-1$ and $C(x)=C_+$ for $x\ge1$, and $(C_-)_{1,1}=(C_+)_{1,1}$.

\section{Scattering theory in a two-Hilbert spaces setting}\label{sec_scattering}
\setcounter{equation}{0}

We discuss in this section the existence and the completeness under smooth
perturbations of the local wave operators for unitary operators in a two-Hilbert
spaces setting. Namely, given two unitary operators $U_0,U$ in Hilbert spaces
$\H_0,\H$ with spectral measures $E^{U_0},E^U$, a bounded operator $J:\H_0\to\H$, and
an open set $\Theta\subset\T:=\{z\in\C\mid|z|=1\}$, we give criteria for the existence
and the completeness of the strong limits
$$
W_\pm(U,U_0,J,\Theta):=\slim_{n\to\pm\infty}U^{-n}JU_0^n E^{U_0}(\Theta)
$$
under the assumption that the difference $JU_0-UJ$ factorises as a product of a
locally $U$-smooth operator on $\Theta$ and a locally $U_0$-smooth operator on
$\Theta$. We start with a standard result on the intertwining property of wave
operators. Note that we use the notation $\B(\H_1,\H_2)$ (resp. $\KK(\H_1,\H_2)$) for
the set of bounded (resp. compact) operators between Hilbert spaces $\H_1$ and $\H_2$,
and we set $\B(\H_1):=\B(\H_1,\H_1)$ and $\KK(\H_1):=\KK(\H_1,\H_1)$.

\begin{Lemma}[Intertwining property]\label{lemma_intertwinning}
Let $U_0,U$ be unitary operators in Hilbert spaces $\H_0,\H$ with spectral measures
$E^{U_0}$, $E^U$, let $J\in\B(\H_0,\H)$, and let $\Theta\subset\T$ be an open set.
Assume that $W_\pm(U,U_0,J,\Theta)$ exist. Then, we have for each bounded Borel
function $\eta:\T\to\C$ the intertwining property
\begin{equation}\label{eq_intertwining}
W_\pm(U,U_0,J,\Theta)\;\!\eta(U_0)=\eta(U)\;\!W_\pm(U,U_0,J,\Theta).
\end{equation}
\end{Lemma}

\begin{proof}
A direct calculation implies the equality
$W_\pm(U,U_0,J,\Theta)\;\!U_0^k=U^k\;\!W_\pm(U,U_0,J,\Theta)$ for each $k\in\Z$. Using
Stone-Weierstrass theorem we infer from this equality that \eqref{eq_intertwining}
holds for each $\eta\in C(\T)$. Finally, using a standard approximation argument in
the weak topology we extend the result to each bounded Borel function $\eta:\T\to\C$.
\end{proof}

Next, we define the closed subspaces of $\H$
$$
\NN_\pm(U,J,\Theta)
:=\left\{\varphi\in\H\mid\lim_{n\to\pm\infty}
\big\|J^*U^nE^U(\Theta)\varphi\big\|_{\H_0}=0\right\},
$$
and note that $E^U(\T\setminus\Theta)\H\subset\NN_\pm(U,J,\Theta)$, that $U$ is
reduced by $\NN_\pm(U,J,\Theta)$, and that
$$
\overline{\Ran\big(W_\pm(U,U_0,J,\Theta)\big)}\perp\NN_\pm(U,J,\Theta),
$$
this last fact being shown as in the self-adjoint case, see \cite[Lemma~3.2.1]{Yaf92}.
In particular, one has the inclusion
$$
\overline{\Ran\big(W_\pm(U,U_0,J,\Theta)\big)}
\subset E^U(\Theta)\H\ominus\NN_\pm(U,J,\Theta),
$$
which motivates the following definition:

\begin{Definition}[$J$-completeness]\label{def_J_complete}
Assume that $W_\pm(U,U_0,J,\Theta)$ exist. The operators $W_\pm(U,U_0,J,\Theta)$ are
$J$-complete on $\Theta$ if
$$
\overline{\Ran\big(W_\pm(U,U_0,J,\Theta)\big)}
=E^U(\Theta)\H\ominus\NN_\pm(U,J,\Theta).
$$
\end{Definition}

\begin{Remark}
In the particular case $\H_0=\H$ and $J=1_\H$, the $J$-completeness on $\Theta$ reduces
to the completeness of $W_\pm(U,U_0,J,\Theta)$ on $\Theta$ in the usual sense. Namely,
$\overline{\Ran\big(W_\pm(U,U_0,1_\H,\Theta)\big)}=E^U(\Theta)\H$, and the operators
$W_\pm(U,U_0,1_\H,\Theta)$ are unitary from $E^{U_0}(\Theta)\H$ to $E^U(\Theta)\H$.
\end{Remark}

The following criterion for $J$-completeness is shown as in the self-adjoint case, see
for example \cite[Thm.~3.2.4]{Yaf92}:

\begin{Lemma}
If $W_\pm(U,U_0,J,\Theta)$ and $W_\pm(U_0,U,J^*,\Theta)$ exist, then
$W_\pm(U,U_0,J,\Theta)$ are $J$-complete on $\Theta$.
\end{Lemma}

\begin{proof}
The intertwining property and the existence of the operators $W_\pm(U_0,U,J^*,\Theta)$
imply that for any $\varphi\in\H$ and $\psi\in\H_0$
\begin{align*}
\big\langle W_\pm(U,U_0,J,\Theta)^*\varphi,\psi\big\rangle_{\H_0}
&=\big\langle\varphi,E^U(\Theta)W_\pm(U,U_0,J,\Theta)\psi\big\rangle_\H\\
&=\lim_{n\to\pm\infty}\big\langle E^U(\Theta)\varphi,
U^{-n}JU_0^n E^{U_0}(\Theta)\psi\big\rangle_\H\\
&=\lim_{n\to\pm\infty}\big\langle E^{U_0}(\Theta)U_0^{-n}J^*U^nE^{U}(\Theta)\varphi,
\psi\rangle_{\H_0}\\
&=\big\langle W_\pm(U_0,U,J^*,\Theta)\varphi,\psi\big\rangle_{\H_0}.
\end{align*}
Thus, $W_\pm(U_0,U,J^*,\Theta)$ is the adjoint of $W_\pm(U,U_0,J,\Theta)$. Since
$\ker\big(W_\pm(U_0,U,J^*,\Theta)\big)=\NN_\pm(U,J,\Theta)$ and
$E^U(\T\setminus\Theta)\H\subset\NN_\pm(U,J,\Theta)$, it follows that
\begin{align*}
\overline{\Ran\big(W_\pm(U,U_0,J,\Theta)\big)}
&=\H\ominus\ker\big(W_\pm(U,U_0,J,\Theta)^*\big)\\
&=\H\ominus\NN_\pm(U,J,\Theta)\\
&=E^U(\Theta)\H\ominus\NN_\pm(U,J,\Theta),
\end{align*}
which proves the claim.
\end{proof}

For the next theorem, we recall that the spectral support $\supp_U(\varphi)$ of a
vector $\varphi\in\H$ with respect to $U$ is the smallest closed set $\Omega\subset\T$
such that $E^U(\Omega)\varphi=\varphi$. We also recall that if $\G$ is an auxiliary
Hilbert space, then an operator $T\in\B(\H,\G)$ is locally $U$-smooth on an open set
$\Theta\subset\T$ if for each closed set $\Theta'\subset\Theta$ there exists
$c_{\Theta'}\ge0$ such that
\begin{equation}\label{def_U_smooth}
\sum_{n\in\Z}\big\|TU^nE^U(\Theta')\varphi\big\|_\G^2
\le c_{\Theta'}\;\!\|\varphi\|_\H^2\quad\hbox{for each $\varphi\in\H$},
\end{equation}
and we refer to \cite[Sec.~2]{FRT13} or \cite[Sec.~3.1]{RST_1} for more information on
locally $U$-smooth operators.

\begin{Theorem}\label{thm_wave}
Let $U_0,U$ be unitary operators in Hilbert spaces $\H_0,\H$ with spectral measures
$E^{U_0},E^U$, $J\in\B(\H_0,\H)$, and $\Theta\subset\T$ be an open set. Let $\G$ be an
auxiliary Hilbert space, $T_0\in\B(\H_0,\G)$ a locally $U_0$-smooth operator on
$\Theta$ and $T\in\B(\H,\G)$ a locally $U$-smooth operator on $\Theta$ such that
$JU_0-UJ=T^*T_0$. Then, the wave operators
\begin{equation}\label{eq_wave}
W_\pm(U,U_0,J,\Theta)=\slim_{n\to\pm\infty}U^{-n}JU_0^n E^{U_0}(\Theta)
\end{equation}
exist, are $J$-complete on $\Theta$, and satisfy the relations
$$
W_\pm(U,U_0,J,\Theta)^*=W_\pm(U_0,U,J^*,\Theta)
\quad\hbox{and}\quad
W_\pm(U,U_0,J,\Theta)\eta(U_0)=\eta(U)W_\pm(U,U_0,J,\Theta)
$$
for each bounded Borel function $\eta:\T\to\C$.
\end{Theorem}

\begin{proof}
We adapt the proof of \cite[Thm.~7.1.4]{ABG96} to the case of unitary operators in a
two-Hilbert spaces setting. The existence of the limits \eqref{eq_wave} is a direct
consequence of the following assertion: For each $\varphi_0\in\H_0$ such that
$\Theta_0:=\supp_{U_0}(\varphi_0)\subset\Theta$, and for each
$\eta\in C_{\rm c}^\infty(\Theta,\R)$ such that $\eta(\theta)=1$ on a neighbourhood of
$\Theta_0$
\begin{equation}\label{eq_equiv_wave}
\slim_{n\to\pm\infty}\eta(U)U^{-n}JU_0^n\varphi_0~\hbox{exist}
\quad\hbox{and}\quad
\lim_{n\to\pm\infty}\big\|\big(1-\eta(U)\big)U^{-n}JU_0^n\varphi_0\big\|_\H=0.
\end{equation}
To prove the first claim in \eqref{eq_equiv_wave}, we take $\varphi\in\H$ and observe
that $W_n:=\eta(U)U^{-n}JU_0^n$ satisfies for $m\le n-1$
\begin{align*}
\Big|\big\langle\varphi,(W_n-W_m)\varphi_0\big\rangle_\H\Big|
&=\left|\,\sum_{j=m+1}^n\big\langle\varphi,\eta(U)U^{-j}(JU_0-UJ)
U_0^{j-1}\varphi_0\big\rangle_\H\right|\\
&=\left|\,\sum_{j=m+1}^n\big\langle TU^{j}\eta(U)\varphi,
T_0U_0^{j-1}\varphi_0\big\rangle_\G\right|\\
&\le\left(\sum_{j=m+1}^n\big\|TU^{j}\eta(U)\varphi\big\|^2_\G\right)^{1/2}
\left(\sum_{j=m+1}^n\big\|T_0U_0^{j-1}\varphi_0\big\|^2_\G\right)^{1/2}\\
&\le c_{\Theta_1}^{1/2}\;\!\|\varphi\|_\H
\left(\sum_{j=m+1}^n\big\|T_0U_0^{j-1)}\varphi_0\big\|^2_\G\right)^{1/2},
\end{align*}
with $\Theta_1:=\supp(\eta)$ and $c_{\Theta_1}$ the constant appearing in the
definition \eqref{def_U_smooth} of a locally $U$-smooth operator. Since $T_0$ is
locally $U_0$-smooth on $\Theta$, it follows that $\|(W_n-W_m)\varphi_0\|_\H\to0$ as
$m\to\infty$ or $n\to-\infty$. This proves the first claim in \eqref{eq_equiv_wave}.

To prove the second claim in \eqref{eq_equiv_wave}, we take
$\eta_0\in C_{\rm c}^\infty(\Theta,\R)$ such that $\eta_0\equiv1$ on $\Theta_0$ and
$\eta\eta_0=\eta_0$. Then, we have $\varphi_0=\eta_0(U_0)\varphi_0$ and
$$
\big(1-\eta(U)\big)J\eta_0(U_0)=\big(1-\eta(U)\big)\big(J\eta_0(U_0)-\eta_0(U)J\big),
$$
and thus the second claim in \eqref{eq_equiv_wave} follows from
$$
\lim_{n\to\pm\infty}\big\|\big(J\eta_0(U_0)-\eta_0(U)J\big)U_0^n\varphi_0\big\|_\H=0.
$$
Since the set of monomials $z^k$ with $z\in\T$ and $k\in\Z$ is total in $C(\T)$ for
the $\sup$ norm, it is sufficient to show that
$$
\lim_{n\to\pm\infty}\big\|\big(JU_0^k-U^kJ\big)U_0^n\varphi_0\big\|_\H=0
$$
for all $k\in\Z$. For $k\ge1$, the result follows from the formula
$
JU_0^k-U^kJ=\sum_{j=1}^kU^{j-1}(JU_0-UJ)U_0^{k-j}
$
and the local $U_0$-smoothness of $T_0$ since
\begin{align*}
\lim_{n\to\pm\infty}\big\|\big(JU_0^k-U^kJ\big)U_0^n\varphi_0\big\|_\H
&\le\lim_{n\to\pm\infty}\sum_{j=1}^k
\big\|(JU_0-UJ)U_0^{k-j}U_0^n\varphi_0\big\|_\H\\
&\le{\rm Const.}\;\!k\;\!\lim_{m\to\pm\infty}
\big\|T_0U_0^m\eta_0(U_0)\varphi_0\big\|_\G\\
&=0,
\end{align*}
and for $k\le0$ the result follows from what precedes since
$
JU_0^{-|k|}-U^{-|k|}J=-U^{-|k|}\big(JU_0^{|k|}-U^{|k|}J\big)U_0^{-|k|}
$.

So, the existence of the limits \eqref{eq_wave} has been established. Similar
arguments, using the relation $U_0^*J^*-J^*U^*=T_0^*T$ instead of $JU_0-UJ=T^*T_0$,
show that $W_\pm(U_0,U,J^*,\Theta)$ exist too. This, together with standard arguments
in scattering theory, implies the claims that follow \eqref{eq_wave}.
\end{proof}

To present the last result of this section, we need to recall some basic definitions
of the conjugate operator theory borrowed from \cite[Chap.~5]{ABG96}: Let $S\in\B(\H)$
and let $A$ be a self-adjoint operator in $\H$ with domain $\dom(A)$. For $k\in\N$, we
say that $S\in C^k(A)$ if the map $\R\ni t\mapsto\e^{-itA}S\e^{itA}\in\B(\H)$ is
strongly of class $C^k$. In the case $k=1$, one has $S\in C^1(A)$ if and only if the
quadratic form
$$
\dom(A)\ni\varphi\mapsto\big\langle A\;\!\varphi,S\varphi\big\rangle_\H
-\big\langle\varphi,SA\;\!\varphi\big\rangle_\H\in\C
$$
is continuous for the topology induced by $\H$ on $\dom(A)$. The operator associated
to the continuous extension of the form is denoted by $[A,S]\in\B(\H)$. Three
regularity conditions slightly stronger than $S\in C^1(A)$ are defined as follows:
$S\in C^{1,1}(A)$ if
$$
\int_0^1\big\|\e^{-itA}S\e^{itA}+\e^{itA}S\e^{-itA}-2S\big\|_{\B(\H)}
\,\frac{\d t}{t^2}<\infty.
$$
$S\in C^{1+0}(A)$ if $S\in C^1(A)$ and
$$
\int_0^1\big\|\e^{-itA}[A,S]\e^{itA}-[A,S]\big\|_{\B(\H)}\,\frac{\d t}t<\infty.
$$
$S\in C^{1+\varepsilon}(A)$ for some $\varepsilon\in(0,1)$ if $S\in C^1(A)$ and
$$
\big\|\e^{-itA}[A,S]\e^{itA}-[A,S]\big\|_{\B(\H)}
\le{\rm Const.}\;\!t^\varepsilon\quad\hbox{for all $t\in(0,1)$.}
$$
As banachisable topological vector spaces, these sets satisfy the continuous
inclusions \cite[Sec.~5.2.4]{ABG96}
$$
C^2(A)\subset C^{1+\varepsilon}(A)\subset C^{1+0}(A)\subset C^{1,1}(A)\subset C^1(A)
\subset C^0(A).
$$

Let us also recall from \cite[Sec.~3.1]{RST_1} that if $U$ is unitary operator in $\H$
with $U\in C^1(A)$, then the function $\widetilde\varrho^A_U:\T\to(-\infty,\infty]$ is
defined by
$$
\widetilde\varrho^A_U(\theta)
:=\sup\big\{a\in\R \mid\exists\;\!\varepsilon>0~\hbox{such that}
~E^U(\theta;\varepsilon)U^{-1}[A,U]E^U(\theta;\varepsilon)\gtrsim
a\;\!E^U(\theta;\varepsilon)\big\},\quad\theta\in\T,
$$
where $E^U(\theta;\varepsilon):=E^U\big(\Theta(\theta;\varepsilon)\big)$,
$\Theta(\theta;\varepsilon):=\{\theta'\in\T\mid|\arg(\theta-\theta')|<\varepsilon\}$,
and for $S,T\in\B(\H)$ the notation $T\gtrsim S$ means that there exists $K\in\KK(\H)$
such that $T+K\ge S$. By analogy with the self-adjoint case, we say that $A$ is
conjugate to $U$ at a point $\theta\in\T$ if $\widetilde\varrho^A_U(\theta)>0$, and we
write
$$
\widetilde\mu^A(U):=\big\{\theta\in\T\mid\widetilde\varrho^A_U(\theta)>0\big\}
$$
for the open subset of $\T$ where $A$ is conjugate to $U$. The set
$\widetilde\mu^A(U)$ is open because the function $\widetilde\varrho^A_U(\theta)$ is
lower semicontinuous. Finally, we denote by $\sigma_{\rm p}(U_0)$ and
$\sigma_{\rm p}(U)$ the pure point spectra of $U_0$ and $U$.

Now, by combining \cite[Thm.~3.4]{RST_1} and Theorem \ref{thm_wave}, we obtain the
following criterion for the existence and completeness of the local wave operators.

\begin{Corollary}\label{corol_wave}
Let $U_0,U$ be unitary operators in Hilbert spaces $\H_0,\H$ with spectral measures
$E^{U_0},E^U$ and $A_0,A$ self-adjoint operators in $\H_0,\H$. Assume either that
$U_0,U$ have a spectral gap and $U_0\in C^{1,1}(A_0),U\in C^{1,1}(A)$, or that
$U_0\in C^{1+0}(A_0),U\in C^{1+0}(A)$. Let
$$
\Theta
:=\big\{\widetilde\mu^{A_0}(U_0)\setminus\sigma_{\rm p}(U_0)\big\}
\cap\big\{\widetilde\mu^{A}(U)\setminus\sigma_{\rm p}(U)\big\},
$$
$J\in\B(\H_0,\H)$, $\G$ be an auxiliary Hilbert space, and assume there exist
$T_0\in\B(\H_0,\G)$ and $T\in\B(\H,\G)$ with $JU_0-UJ=T^*T_0$ and such that $T_0$
extends continuously to an element of $\B\big(\dom(\langle A_0\rangle^s)^*,\G\big)$
and $T$ extends continuously to an element of
$\B\big(\dom(\langle A\rangle^s)^*,\G\big)$ for some $s>1/2$. Then, the strong limits
$$
W_\pm(U,U_0,J,\Theta):=\slim_{n\to\pm\infty}U^{-n}JU_0^nE^{U_0}(\Theta)
$$
exist, are $J$-complete on $\Theta$, and satisfy the relations
$$
W_\pm(U,U_0,J,\Theta)^*=W_\pm(U_0,U,J^*,\Theta)
\quad\hbox{and}\quad
W_\pm(U,U_0,J,\Theta)\;\!\eta(U_0)=\eta(U)\;\!W_\pm(U,U_0,J,\Theta)
$$
for each bounded Borel function $\eta:\T\to\C$.
\end{Corollary}

\section{Scattering theory for quantum walks with an anisotropic coin}
\label{sec_walks}
\setcounter{equation}{0}

In this section, we present our results on the scattering theory for the pair
$\{U_0,U\}$ when $U$ is the evolution operator of a one-dimensional quantum walk with
an anisotropic coin and $U_0$ is the corresponding free evolution operator. We start
by recalling from \cite[Sec.~4]{RST_1} the needed definitions and facts on the
operators $U$ and $U_0$.

Let $\H$ be the Hilbert space of square-summable $\C^2$-valued sequences
$$
\H:=\ell^2(\Z,\C^2)
=\left\{\Psi:\Z\to\C^2\mid\sum_{x\in\Z}\|\Psi(x)\|_2^2<\infty\right\},
$$
where $\|\cdot\|_2$ is the usual norm on $\C^2$. Then, the evolution operator of the
one-dimensional quantum walk in $\H$ that we consider is given by $U:=SC$, with $S$ a
shift operator defined by
$$
(S\Psi)(x)
:=\begin{pmatrix}
\Psi^{(0)}(x+1)\\
\Psi^{(1)}(x-1)
\end{pmatrix},
\quad
\Psi
=\begin{pmatrix}
\Psi^{(0)}\\
\Psi^{(1)}
\end{pmatrix}\in\H,~x\in\Z,
$$
and $C$ a coin operator defined by
$$
(C\Psi)(x):=C(x)\Psi(x),\quad \Psi\in\H,~x\in\Z,~C(x)\in\U(2).
$$
In particular, the evolution operator $U$ is unitary in $\H$ since both $S$ and $C$
are unitary in $\H$.

We assume that the coin operator $C$ has an anisotropic behaviour at infinity. More
precisely, we assume that $C$ converges with short-range rate to two asymptotic coin
operators, one on the left and one on the right in the following way:

\begin{Assumption}[Short-range assumption]\label{ass_short}
There exist $C_\ell,C_{\rm r}\in\U(2)$, $\kappa_\ell,\kappa_{\rm r}>0$, and
$\varepsilon_\ell,\varepsilon_{\rm r}>0$ such that
\begin{align*}
&\big\|C(x)-C_\ell\big\|_{\B(\C^2)}
\le\kappa_\ell\;\!|x|^{-1-\varepsilon_\ell}\quad\mathrm{if}~x<0\\
&\big\|C(x)-C_{\rm r}\big\|_{\B(\C^2)}
\le\kappa_{\rm r}\;\!|x|^{-1-\varepsilon_{\rm r}}\quad\mathrm{if}~x>0,
\end{align*}
where the indexes $\ell$ and ${\rm r}$ stand for ``left" and ``right".
\end{Assumption}

This assumption provides us with two new unitary operators $U_\ell:=SC_\ell$ and
$U_{\rm r}:=SC_{\rm r}$ describing the asymptotic behaviour of $U$ on the left and on
the right.

From now on, we shall use the symbol $\star$ to denote either the index $\ell$ or the
index ${\rm r}$, and we define the space
$$
\H_{\rm fin}
:=\bigcup_{n\in\N}\big\{\Psi\in\H\mid\hbox{$\Psi(x)=0$ if $|x|\ge n$}\big\}
\subset\H,
$$
the Hilbert space $\K:=\ltwo\big([0,2\pi),\frac{\d k}{2\pi},\C^2\big)$, and the
unitary Fourier transform $\F:\H\to\K$ which corresponds to the unique continuous
extension of the operator
$$
(\F\Psi)(k):=\sum_{x\in\Z}\e^{-ikx}\Psi(x),\quad\Psi\in\H_{\rm fin},~k\in[0,2\pi).
$$
The operator $U_\star$ is decomposable in the Fourier representation, namely, for all
$f\in\K$ and almost every $k\in[0,2\pi)$ we have
$$
(\F\;\!U_\star\;\!\F^*f)(k)=\widehat{U_\star}(k)f(k)\quad\hbox{with}\quad
\widehat{U_\star}(k)
:=\begin{pmatrix}
\e^{ik}&0\\
0&\e^{-ik}
\end{pmatrix}
C_\star\in\U(2).
$$
Also, since $\widehat{U_\star}(k)\in\U(2)$, the spectral theorem implies that
$$
\widehat{U_\star}(k)=\sum_{j=1}^2\lambda_{\star,j}(k)\;\!\Pi_{\star,j}(k),
$$
with $\lambda_{\star,j}(k)$ the eigenvalues of $\widehat{U_\star}(k)$ and
$\Pi_{\star,j}(k)$ the corresponding orthogonal projections. Furthermore, for
$j\in\{1,2\}$ we let $v_{\star,j}:[0,2\pi)\to\R$ be the bounded function given by
\begin{equation}\label{def_small_v}
v_{\star,j}(k):=i\;\!\lambda_{\star,j}'(k)\big(\lambda_{\star,j}(k)\big)^{-1},
\end{equation}
where $(\;\!\cdot\;\!)'$ means the derivative with respect to $k$. The function
$v_{\star,j}$ is real valued because $\lambda_{\star,j}$ takes values in $\T$. Then,
we define for all $f\in\K$ and almost every $k\in[0,2\pi)$ the decomposable operator
$\widehat{V_\star}\in\B(\K)$,
\begin{equation}\label{def_big_V}
\big(\widehat{V_\star}f\big)(k):=\widehat{V_\star}(k)f(k)
\quad\hbox{where}\quad
\widehat{V_\star}(k):=\sum_{j=1}^2v_{\star,j}(k)\;\!\Pi_{\star,j}(k)\in\B(\C^2),
\end{equation}
and we call asymptotic velocity operator the operator
$V_\star:=\F^*\;\!\widehat{V_\star}\;\!\F$.

We can now start studying the scattering theory for the operator $U$. As free
evolution operator, we use the unitary operator $U_0:=U_\ell\oplus U_{\rm r}$ in the
Hilbert space $\H_0:=\H\oplus\H$. In \cite[Sec.~4.2]{RST_1}, it has been shown that
the spectrum of $U_0$ coincides with the essential spectrum of $U$, namely,
$$
\sigma_{\rm ess}(U)=\sigma(U_\ell)\cup\sigma(U_{\rm r})=\sigma(U_0).
$$
As identification operator between the Hilbert spaces $\H_0$ and $\H$, we use the
operator $J\in\B(\H_0,\H)$ defined by
$$
J(\Psi_\ell,\Psi_{\rm r}):=j_\ell\;\!\Psi_\ell+j_{\rm r}\;\!\Psi_{\rm r},
\quad(\Psi_\ell,\Psi_{\rm r})\in\H_0,
$$
where
$$
j_{\rm r}(x):=
\begin{cases}
1 & \hbox{if $x\ge0$}\\
0 & \hbox{if $x\le-1$}
\end{cases}
\quad\hbox{and}\quad
j_\ell:=1-j_{\rm r}.
$$	

The first lemma of the section consists in a simple observation related to the
$J$-completeness of the wave operators for the pair $\{U_0,U\}:$

\begin{Lemma}\label{lemma_lim_J}
For any $\Psi\in\H$, we have $\lim_{n\to\infty}\|J^*U^n\;\!\Psi\|_{\H_0}=0$ if and
only if $\Psi=0$.
\end{Lemma}

\begin{proof}
We know from \cite[Lemma~4.7]{RST_1} that $JJ^*=1_\H$. Therefore, we have for any
$n\in\Z$ and $\Psi\in\H$
$$
\|J^*U^n\;\!\Psi\|_{\H_0}^2
=\big\langle J^*U^n\;\!\Psi,J^*U^n\;\!\Psi\big\rangle_{\H_0}
=\big\langle U^n\;\!\Psi,JJ^*U^n\;\!\Psi\big\rangle_\H
=\|U^n\;\!\Psi\|_\H^2
=\|\Psi\|_\H^2,
$$
which implies the claim.
\end{proof}

A direct consequence of this lemma is that the abstract spaces $\NN_\pm(U,J,\Theta)$
defined in Section \ref{sec_scattering} are trivial in our case:
\begin{equation}\label{eq_spaces_eta}
\NN_\pm(U,J,\Theta)=E^U(\T\setminus\Theta)\H.
\end{equation}

Now, in order to prove with the help of Corollary \ref{corol_wave} the existence and
the completeness of the wave operators for the pair $\{U_0,U\}$, we need to recall
some facts about conjugate operators $A_0$ and $A$ introduced in \cite{RST_1}. In the
proof of \cite[Thm.~4.5(c)]{RST_1}, it has been shown that there exists for
$\star=\ell,{\rm r}$ an operator $A_\star$, defined in terms of the velocity operator
$V_\star$ and essentially self-adjoint on $\H_{\rm fin}$, such that
$U_\star\in C^2(A_\star)$. In addition, the operator $A_\star$ is conjugate to the
operator $U_\star$ outside the set $\partial\sigma(U_\star)$ of boundary points of
$\sigma(U_\star)$ in $\T$. As a consequence, the operator
$A_0:=A_\ell\oplus A_{\rm r}$ is well-defined and conjugate to $U_0$ on the set
$\sigma(U_0)\setminus\tau(U)$, with
$$
\tau(U):=\partial\sigma(U_\ell)\cup\partial\sigma(U_{\rm r}).
$$
The set $\tau(U)$, which contains at most $8$ values, is called for this reason the
set of thresholds of $U$. In \cite[Lemma~4.9]{RST_1}, it has also been shown that the
operator $JA_0J^*$ is essentially self-adjoint on $\H_{\rm fin}$, with self-adjoint
extension denoted by $A$, and that $A$ is conjugate to $U$ on
$\sigma(U_0)\setminus\tau(U)$.

We also recall a relation between the conjugate operator $A$ and the position operator
$Q$ given by
$$
\big(Q\;\!\Psi\big)(x):=x\;\!\Psi(x),
\quad x\in\Z,~\Psi\in\dom(Q):=\big\{\Psi\in\H\mid\|Q\;\!\Psi\|_\H<\infty\big\}.
$$
This relation has already been used in the proof of \cite[Lemma~4.13]{RST_1}, but we
make it more explicit now. As mentioned in that proof, the operator
$\langle Q\rangle^{-1}A_\star$ defined on $\H_{\rm fin}$ extends continuously to an
element of $\B(\H)$. This implies that $\dom(\langle Q\rangle)\subset\dom(A_\star)$,
and thus that $\dom(\langle Q\rangle)\subset\dom(A)$ due to the equality
\begin{equation}\label{eq_J_A_0_J_star}
A=JA_0J^*=j_\ell\;\!A_\ell\;\!j_\ell+j_{\rm r}\;\!A_{\rm r}\;\!j_{\rm r}
\quad\hbox{on}\quad\H_{\rm fin}.
\end{equation}
Therefore, we obtain by real interpolation the inclusions
\begin{equation}\label{eq_inclusions_Q}
\dom(\langle Q\rangle)^s\subset\dom(\langle A\rangle)^s
\quad\hbox{and}\quad
\B\big(\dom(\langle Q\rangle)^{-s},\H\big)
\subset\B\big(\dom(\langle A\rangle)^{-s},\H\big)
\end{equation}
for each $s>0$.

We can now state our theorem on the $J$-completeness of the wave operators for the
pair $\{U_0,U\}$, with the notation $E^{U}_{\rm ac}$ for the orthogonal projection on
the absolutely continuous subspace of $U$.

\begin{Theorem}\label{thm_complete_walks}
Let
$
\Theta:=\{\sigma(U_\ell)\cup\sigma(U_{\rm r})\}
\setminus\{\tau(U)\cup\sigma_{\rm p}(U)\}
$.
Then, the operators
\begin{equation}\label{eq_wave_theta}
W_\pm(U,U_0,J,\Theta)=\slim_{n\to\pm\infty}U^{-n}JU_0^nE^{U_0}(\Theta)
\end{equation}
exist and satisfy $\overline{\Ran\big(W_\pm(U,U_0,J,\Theta)\big)}=E^{U}_{\rm ac}\H$.
In addition, the relations
$$
W_\pm(U,U_0,J,\Theta)^*=W_\pm(U_0,U,J^*,\Theta)
\quad\hbox{and}\quad
W_\pm(U,U_0,J,\Theta)\;\!\eta(U_0)=\eta(U)\;\!W_\pm(U,U_0,J,\Theta)
$$
hold for each bounded Borel function $\eta:\T\to\C$.
\end{Theorem}

Before the proof, it is convenient to highlight some properties of the projection
$E^{U_0}(\Theta)$. First, let the matrices $C_\star\in\U(2)$ be parameterised as
\begin{equation}\label{eq_param}
C_\star=\e^{i\delta_\star/2}
\begin{pmatrix}
a_\star\e^{i(\alpha_\star-\delta_\star/2)}
& b_\star\e^{i(\beta_\star-\delta_\star/2)}\\
-b_\star\e^{-i(\beta_\star-\delta_\star/2)}
& a_\star\e^{-i(\alpha_\star-\delta_\star/2)}
\end{pmatrix}
\end{equation}
with $a_\star,b_\star\in[0,1]$ satisfying $a_\star^2+b_\star^2=1$, and
$\alpha_\star,\beta_\star,\delta_\star\in(-\pi,\pi]$. Then, recall from
\cite[Lemma~4.1 \& Prop.~4.5]{RST_1} that the operator $U_\star$ has pure point
spectrum with $\sigma(U_\star)\subset\tau(U)$ if $a_\star=0$ and purely absolutely
continuous spectrum with $\partial\sigma(U_\star)\subset\tau(U)$ if $a_\star\in(0,1]$.
Since we also know from \cite[Thm.~2.4]{RST_1} that the number of eigenvalues of $U$
in any closed set $\Theta'\subset\T\setminus\tau(U)$ is finite, we infer that
\begin{equation}\label{eq_projectors}
E^{U_0}(\Theta)=E^{U_0}_{\rm ac}=
\begin{cases}
1_\H\oplus1_\H & \hbox{if $a_\ell,a_{\rm r}\in(0,1]$,}\\
1_\H\oplus0_\H & \hbox{if $a_\ell\in(0,1]$ and $a_{\rm r}=0$,}\\
0_\H\oplus1_\H & \hbox{if $a_\ell=0$ and $a_{\rm r}\in(0,1]$,}\\
0_\H\oplus0_\H & \hbox{if $a_\ell=a_{\rm r}=0$.}
\end{cases}
\end{equation}
Thus, in the generic case $a_\ell,a_{\rm r}\in(0,1]$, the projection $E^{U_0}(\Theta)$
appearing in \eqref{eq_wave_theta} can simply be replaced by
$1_{\H_0}=1_\H\oplus1_\H$.

\begin{proof}[Proof of Theorem \ref{thm_complete_walks}]
All the claims except the equality
$\overline{\Ran\big(W_\pm(U,U_0,J,\Theta)\big)}=E^{U}_{\rm ac}\H$ follow from
Corollary \ref{corol_wave} whose assumptions are checked now.

The proof of \cite[Prop.~4.5(c)]{RST_1} implies that $U_0\in C^2(A_0)$,
\cite[Lemma~4.13]{RST_1} implies that $U\in C^{1+\varepsilon}(A)$ for each
$\varepsilon\in(0,1)$ with
$\varepsilon\le\min\{\varepsilon_\ell,\varepsilon_{\rm r}\}$, and
\cite[Prop.~4.11]{RST_1} implies that
$$
\{\sigma(U_\ell)\cup\sigma(U_{\rm r})\}
\setminus\{\tau(U)\cup\sigma_{\rm p}(U)\}
\subset\big\{\widetilde\mu^{A_0}(U_0)\setminus\sigma_{\rm p}(U_0)\big\}
\cap\big\{\widetilde\mu^{A}(U)\setminus\sigma_{\rm p}(U)\big\}.
$$
Thus, in order to apply Corollary \ref{corol_wave}, it is sufficient to prove the
existence of operators $T_0\in\B(\H_0,\G)$ and $T\in\B(\H,\G)$ with $JU_0-UJ=T^*T_0$
and such that $T_0$ extends continuously to an element of
$\B\big(\dom(\langle A_0\rangle^s)^*,\G\big)$ and $T$ extends continuously to an
element of $\B\big(\dom(\langle A\rangle^s)^*,\G\big)$ for some $s>1/2$. For that
purpose, we set $s:=(1+\varepsilon)/2$ and define the sesquilinear form
$D:\H_0\times\H\to\C$ by
$$
D\big((\Psi_\ell,\Psi_{\rm r}),\Psi\big)
:=\left\langle\langle Q\rangle^s\;\!\Psi,\sum_{\star\in\{\ell,{\rm r}\}}
\big([j_\star,S]C_\star-S(C-C_\star)\;\!j_\star\big)
\langle Q\rangle^s\;\!\Psi_\star \right\rangle_\H
$$
for each $(\Psi_\ell,\Psi_{\rm r})\in\H_{\rm fin }\oplus\H_{\rm fin}$ and
$\Psi\in\H_{\rm fin}$. With arguments similar to the ones used in the proofs of
\cite[Lemmas~4.12 \& 4.13]{RST_1}, one shows that the form $D$ extends continuously to
a bounded form on $\H_0\times\H$. Thus, there exists an operator $D\in\B(\H_0,\H)$
(the same notation is used on purpose) such that
$$
D\big((\Psi_\ell,\Psi_{\rm r}),\Psi\big)
=\big\langle \Psi,D(\Psi_\ell,\Psi_{\rm r})\big\rangle_\H,
\quad(\Psi_\ell,\Psi_{\rm r})\in\H_0,~\Psi\in\H.
$$
Also, we define the operators
$T_0:=\langle Q\rangle^{-s}\oplus\langle Q\rangle^{-s}\in\B(\H_0)$ and
$T:=D^*\langle Q\rangle^{-s}\in\B(\H_0,\H)$, and observe that $JU_0-UJ=T^*T_0$ due to
the definition of $D$ and Equation \cite[Eq.~(4.6)]{RST_1}:
$$
JU_0-UJ=\sum_{\star\in\{\ell,{\rm r}\}}
\big([j_\star,S]C_\star-S(C-C_\star)\;\!j_\star\big).
$$
Finally, we note that the second inclusion in \eqref{eq_inclusions_Q} implies that
$$
\langle Q\rangle^{-s}
\in\B\big(\dom(\langle Q\rangle)^{-s},\H\big)
\subset\B\big(\dom(\langle A\rangle)^{-s},\H\big),
$$
and thus that $T\in\B\big(\dom(\langle A\rangle)^{-s},\H_0\big)$. Similarly, since
$$
\langle Q\rangle^{-s}\oplus\langle Q\rangle^{-s}
\in\B\big(\dom(\langle Q\rangle^{-s}\oplus\langle Q\rangle^{-s}),\H_0\big)
\subset\B\big(\dom(\langle A_0\rangle)^{-s},\H_0\big),
$$
we have that $T_0\in\B\big(\dom(\langle A_0\rangle)^{-s},\H_0\big)$, and thus all the
assumptions of Corollary \ref{corol_wave} are verified.

Therefore, it only remains to show that
$\overline{\Ran\big(W_\pm(U,U_0,J,\Theta)\big)}=E^U_{\rm ac}\H$. For this, we recall
from \eqref{eq_spaces_eta} that $\NN_\pm(U,J,\Theta)=E^U(\T\setminus\Theta)\H$. This,
together with the $J$-completeness of the wave operators and \cite[Thm.~2.4]{RST_1},
implies that
$$
\overline{\Ran\big(W_\pm(U,U_0,J,\Theta)\big)}
=E^U(\Theta)\H\ominus\NN_\pm(U,J,\Theta)
=E^U(\Theta)\H=E_{\rm ac}^U \H.
$$
\end{proof}

In the last proposition of the section, we determine explicitly the kernels of the
wave operators $W_\pm(U,U_0,J,\Theta)$. We use the notation $\chi_\Lambda$ for the
characteristic function of a set $\Lambda\subset\R$ and $\chi_\pm$ for the
characteristic functions of the sets $(0,\infty)$ and $(-\infty,0)$, respectively.

\begin{Proposition}\label{prop_kernels}
Let
$
\Theta:=\{\sigma(U_\ell)\cup\sigma(U_{\rm r})\}
\setminus\{\tau(U)\cup\sigma_{\rm p}(U)\}
$.
Then, the wave operators $W_\pm(U,U_0,J,\Theta):\H_0\to\H$ are partial isometries with
initial sets
\begin{equation}\label{def_subspaces}
\H_0^\pm:=\chi_\mp(V_\ell)\H\oplus\chi_\pm(V_{\rm r})\H.
\end{equation}
\end{Proposition}

Let us make two observations before giving the proof. Firstly, if $a_\ell\ne0$ and
$a_{\rm r}\ne0$, then \cite[Lemma~4.6]{RST_1} implies that the value $0$ is not an
atom of the spectral measures of $V_\ell$ and $V_{\rm r}$. Therefore, one has the
following orthogonal decomposition of $\H_0:$
$$
\H_0
=\H\oplus\H
=\big(\chi_-(V_\ell)+\chi_+(V_\ell)\big)\H
\oplus\big(\chi_+(V_{\rm r})+\chi_-(V_{\rm r})\big)\H
=\H_0^+\oplus\H_0^-.
$$
Secondly, if $a_\ell=0$, then \cite[Lemma~4.2(a)]{RST_1} implies that $V_\ell=0$. Thus
$\chi_\pm(V_\ell)=0$, and \eqref{def_subspaces} implies that $W_\pm(U,U_0,J,\Theta)$
are isometric only on vectors $0\oplus\Psi_{\rm r}$ with
$\Psi_{\rm r}\in\chi_\pm(V_{\rm r})\H$. Such a result is not surprising since we know
from \cite[Lemma~4.1(a)]{RST_1} that in this case one has
$\sigma(U_\ell)=\sigma_{\rm p}(U_\ell)$ and
$$
W_\pm(U,U_0,J,\Theta)
=W_\pm(U,U_0,J,\Theta)E^{U_0}(\Theta)
=W_\pm(U,U_0,J,\Theta)\big(0\oplus E^{U_{\rm r}}(\Theta)\big).
$$
A similar result holds if $a_{\rm r}=0$.

\begin{proof}[Proof of proposition \ref{prop_kernels}]
We give the proof for $W_+(U,U_0,J,\Theta)$, treating separately the cases
corresponding to the different values of $a_\ell$ and $a_{\rm r}$. The proof for
$W_-(U,U_0,J,\Theta)$ is similar.

(i) If $a_\ell=a_{\rm r}=0$, then we know from \eqref{eq_projectors} that
$E^{U_0}(\Theta)=0$, and
$$
W_+(U,U_0,J,\Theta)
=W_+(U,U_0,J,\Theta)E^{U_0}(\Theta)
=0.
$$
Thus, the wave operator $W_+(U,U_0,J,\Theta)$ is isometric only on the subspace
$\{0\}\oplus\{0\}$. But, we also have $\H_0^+=\{0\}\oplus\{0\}$ since
$V_\ell=0=V_{\rm r}$. So, $W_+(U,U_0,J,\Theta)$ is a partial isometry with (trivial)
initial set $\H_0^+$.

(ii) If $a_\ell,a_{\rm r}\in(0,1]$, we know from \eqref{eq_projectors} that
$E^{U_0}(\Theta)=1_{\H_0}$. To show that
$\H_0^-\subset\ker\big(W_+(U,U_0,J,\Theta)\big)$, take
$(\Psi_\ell,\Psi_{\rm r})\in\H_0^-$ such that
$\chi_{[\varepsilon,\infty)}(V_\ell)\Psi_\ell=\Psi_\ell$ and
$\chi_{(-\infty,-\varepsilon]}(V_{\rm r})\Psi_{\rm r}=\Psi_{\rm r}$ for some
$\varepsilon>0$. Then, one has
\begin{align*}
\big\|W_+(U,U_0,J,\Theta)(\Psi_\ell,\Psi_{\rm r})\big\|_\H
&=\left\|\slim_{n\to\infty}U^{-n}JU_0^n(\Psi_\ell,\Psi_{\rm r})\right\|_\H\\
&=\lim_{n\to\infty}\left\|\sum_{\star\in\{\ell,{\rm r}\}}
U^{-n}j_\star U^n_\star\;\!\Psi_\star\right\|_\H\\
&\le\sum_{\star\in\{\ell,{\rm r}\}}\lim_{n\to\infty}
\big\|U^{-n}j_\star U^n_\star\;\!\Psi_\star\big\|_\H\\
&=\sum_{\star\in\{\ell,{\rm r}\}}\lim_{n\to\infty}
\big\|U^{-n}_\star j_\star U^n_\star\;\!\Psi_\star\big\|_\H.
\end{align*}
Now, if $\eta_\ell,\eta_{\rm r}\in C(\R,[0,1])$ satisfy
$$
\eta_\ell(s):=
\begin{cases}
1 & \hbox{if $s<0$}\\
0 & \hbox{if $s\ge\varepsilon$}\\
\end{cases}
\quad\hbox{and}\quad
\eta_{\rm r}(s):=
\begin{cases}
0 & \hbox{if $s\le-\varepsilon$}\\
1 & \hbox{if $s>0$,}\\
\end{cases}
$$
one obtains for each $n\in\N^*$ the inequality
$$
\big\|U^{-n}_\star j_\star U^n_\star\;\!\Psi_\star\big\|_\H
\le\big\|U^{-n}_\star\;\!\eta_\star(Q/n)U^n_\star
\;\!\Psi_\star\big\|_\H.
$$
Furthermore, since
$$
\eta_\ell(V_\ell)\Psi_\ell
=\eta_\ell(V_\ell)\;\!\chi_{[\varepsilon,\infty)}(V_\ell)\Psi_\ell
=0
\quad\hbox{and}\quad
\eta_{\rm r}(V_{\rm r})\Psi_{\rm r}
=\eta_{\rm r}(V_{\rm r})\;\!\chi_{(-\infty,-\varepsilon]}(V_{\rm r})\Psi_{\rm r}
=0,
$$
one infers from \cite[Thm.~4,1]{Suz16} and a standard result on the convergence in the
strong resolvent sense \cite[Thm.~VIII.20(b)]{RS1} that
$$
\lim_{n\to\infty}\big\|U^{-n}_\star\;\!\eta_\star(Q/n)U^n_\star
\;\!\Psi_\star\big\|_\H
=\big\|\;\!\eta_\star(V_\star)\Psi_\star\big\|_\H
=0.
$$
Putting together what precedes, one obtains that
\begin{align*}
\big\|W_+(U,U_0,J,\Theta)(\Psi_\ell,\Psi_{\rm r})\big\|_\H
\le\sum_{\star\in\{\ell,{\rm r}\}}\lim_{n\to\infty}
\big\|U^{-n}_\star j_\star U^n_\star\;\!\Psi_\star\big\|_\H
=0,
\end{align*}
meaning that $(\Psi_\ell,\Psi_{\rm r})\in\ker\big(W_+(U,U_0,J,\Theta)\big)$. Since
$$
(\Psi_\ell,\Psi_{\rm r})
=\big(\chi_{[\varepsilon,\infty)}(V_\ell)\Psi_\ell,
\chi_{(-\infty,-\varepsilon]}(V_{\rm r})\Psi_{\rm r}\big),
$$
a density argument taking into account the fact that the value $0$ is not an atom of
the spectral measures of $V_\ell$ and $V_{\rm r}$ then shows that
$\H_0^-\subset\ker\big(W_+(U,U_0,J,\Theta)\big)$.

To show that $W_+(U,U_0,J,\Theta)$ is an isometry on $\H_0^+$, take
$(\Psi_\ell,\Psi_{\rm r})\in\H_0^+$ such that
$\chi_{(-\infty,-\varepsilon]}(V_\ell)\Psi_\ell=\Psi_\ell$ and
$\chi_{[\varepsilon,\infty)}(V_{\rm r})\Psi_{\rm r}=\Psi_{\rm r}$ for some
$\varepsilon>0$, and let $\zeta_\ell,\zeta_{\rm r}\in C(\R,[0,1])$ satisfy
$$
\zeta_\ell(s):=
\begin{cases}
0 & \hbox{if $s\le-\varepsilon$}\\
1 & \hbox{if $s>0$}\\
\end{cases}
\quad\hbox{and}\quad
\zeta_{\rm r}(s):=
\begin{cases}
1 & \hbox{if $s<0$}\\
0 & \hbox{if $s\ge\varepsilon$.}\\
\end{cases}
$$
Then, using successively the identity $E^{U_0}(\Theta)=1_{\H_0}$, the identity
$J^*J=j_\ell\oplus j_{\rm r}$ of \cite[Lemma~4.7]{RST_1}, the definition of the
asymptotic velocity $V_\star$, and the assumption on the support of $\Psi_\star$, one
gets
\begin{align*}
\left|\big\|W_+(U,U_0,J,\Theta)(\Psi_\ell,\Psi_{\rm r})\big\|_\H^2
-\big\|(\Psi_\ell,\Psi_{\rm r})\big\|_{\H_0}^2\right|
&=\lim_{n\to\infty}\left|\big\|U^{-n}JU_0^n(\Psi_\ell,\Psi_{\rm r})\big\|_\H^2
-\big\|(\Psi_\ell,\Psi_{\rm r})\big\|_{\H_0}^2\right|\\
&=\lim_{n\to\infty}\left|\big\langle U_0^n(\Psi_\ell,\Psi_{\rm r}),
(J^*J-1)U_0^n(\Psi_\ell,\Psi_{\rm r})\big\rangle_{\H_0}\right|\\
&=\lim_{n\to\infty}\big\langle U_0^n(\Psi_\ell,\Psi_{\rm r}),
(1-j_\ell\oplus j_{\rm r})U_0^n(\Psi_\ell,\Psi_{\rm r})\big\rangle_{\H_0}\\
&\le\sum_{\star\in\{\ell,{\rm r}\}}\lim_{n\to\infty}\big\langle\Psi_\star,
U_\star^{-n}\zeta_\star(Q/n)U_\star^n\;\!\Psi_\star\big\rangle_\H\\
&=\sum_{\star\in\{\ell,{\rm r}\}}\big\langle\Psi_\star,
\zeta_\star(V_\star)\Psi_\star\big\rangle_\H\\
&=0.
\end{align*}
Thus, $W_+(U,U_0,J,\Theta)$ is isometric on $(\Psi_\ell,\Psi_{\rm r})$. Since
$$
(\Psi_\ell,\Psi_{\rm r})
=\big(\chi_{(-\infty,-\varepsilon]}(V_\ell)\Psi_\ell,
\chi_{[\varepsilon,\infty)}(V_{\rm r})\Psi_{\rm r}\big),
$$
a density argument taking into account the fact that the value $0$ is not an atom of
the spectral measures of $V_\ell$ and $V_{\rm r}$ then shows that
$W_+(U,U_0,J,\Theta)$ is an isometry on whole of $\H_0^+$.

(iii) If $a_\ell=0$ and $a_{\rm r}\in(0,1]$ or if $a_\ell\in(0,1]$ and $a_{\rm r}=0$,
then the claim can be shown as in point (ii). We leave the details to the reader.
\end{proof}

\begin{Remark}\label{rem_sum_wave}
Let $\Psi=(\Psi_\ell,\Psi_{\rm r})\in\H_0$. Then, we have
\begin{align*}
W_\pm(U,U_0,J,\Theta)\Psi
&=\slim_{n\to\pm\infty}U^{-n}JU_0^nE^{U_0}(\Theta)\Psi\\
&=\slim_{n\to\pm\infty}U^{-n}\Big(j_\ell U_\ell^nE^{U_\ell}(\Theta)\Psi_\ell
+j_{\rm r}U_{\rm r}^nE^{U_{\rm r}}(\Theta)\Psi_{\rm r}\Big)\\
&=W_\pm(U,U_\ell,j_\ell,\Theta)\Psi_\ell
+W_\pm(U,U_{\rm r},j_{\rm r},\Theta)\Psi_{\rm r}
\end{align*}
with
$$
W_\pm(U,U_\star,j_\star,\Theta)
:=\slim_{n\to\pm\infty}U^{-n}j_\star U_\star^nE^{U_\star}(\Theta).
$$
That is, the wave operators $W_\pm(U,U_0,J,\Theta)$ act as the sum of the operators
$W_\pm(U,U_\star,j_\star,\Theta):$
$$
W_\pm(U,U_0,J,\Theta)\Psi
=W_\pm(U,U_\ell,j_\ell,\Theta)\Psi_\ell
+W_\pm(U,U_{\rm r},j_{\rm r},\Theta)\Psi_{\rm r}.
$$
This simple observation will be used in the following section.
\end{Remark}

\section{Weak limit theorem}\label{sec_weak_lim}
\setcounter{equation}{0}

We prove in this section a weak limit theorem for quantum walks with an anisotropic
coin, and we give an interpretation of this weak limit theorem by comparing it with
its classical analogue, the central limit theorem for classical random walks. Our
first proposition gives a description of the asymptotic velocity operator associated
to the the full evolution operator $U$. To state it, we introduce the Heisenberg
evolution $Q(n):=U^{-n}QU^n$, $n\in\Z$, of the position operator $Q$, and the velocity
operator $V_0:=V_\ell\oplus V_{\rm r}$ for the free evolution operator
$U_0=U_\ell\oplus U_{\rm r}$.

\begin{Proposition}\label{prop_Heisenberg}
Let
$
\Theta:=\{\sigma(U_\ell)\cup\sigma(U_{\rm r})\}
\setminus\{\tau(U)\cup\sigma_{\rm p}(U)\}
$,
$V_{\rm ac}:=W_+(U,U_0,J,\Theta)\;\!V_0\;\!W_+(U,U_0,J,\Theta)^*$, and $\xi\in\R$.
Then, one has
$$
\slim_{n\to\infty}\e^{i\xi Q(n)/n}=E^U_{\rm p}+\e^{i\xi V_{\rm ac}}E^U_{\rm ac}.
$$
\end{Proposition}

\begin{proof}
The finiteness of $\tau(U)$ and \cite[Thm.~2.4]{RST_1} imply that $U$ has no singular
continuous spectrum, and \cite[Thm.~4.2]{Suz16} implies the equality
$\slim_{n\to\infty}\e^{i\xi Q(n)/n}E^U_{\rm p}=E^U_{\rm p}$. Therefore,
$$
\slim_{n\to\infty}\e^{i\xi Q(n)/n}
=\slim_{n\to\infty}\e^{i\xi Q(n)/n}\big(E^U_{\rm p}+E^U_{\rm ac}\big)
=E^U_{\rm p}+\slim_{n\to\infty}\e^{i\xi Q(n)/n}E^U_{\rm ac},
$$
and thus it is sufficient to show that
\begin{equation}\label{eq_to_prove}
\slim_{n\to\infty}\e^{i\xi Q(n)/n}E^U_{\rm ac}=\e^{i\xi V_{\rm ac}}E^U_{\rm ac}.
\end{equation}

Now, we know from Theorem \ref{thm_complete_walks} and Proposition \ref{prop_kernels}
that $\overline{\Ran\big(W_+(U,U_0,J,\Theta)\big)}=E^U_{\rm ac}\H$ and that
$W_+(U,U_0,J,\Theta)$ is a partial isometry with initial set $\H_0^+$. Furthermore, we
have the inclusion $V_0\;\!\H_0^+\subset\H_0^+$ due to the definition of $\H_0^+$ (see
\eqref{def_subspaces}). So, we obtain for each $k\in\N$ the identity
$$
\big(W_+(U,U_0,J,\Theta)\;\!V_0\;\!W_+(U,U_0,J,\Theta)^*\big)^kE^U_{\rm ac}
=W_+(U,U_0,J,\Theta)\;\!V_0^k\;\!W_+(U,U_0,J,\Theta)^*E^U_{\rm ac},
$$
and thus the r.h.s. of \eqref{eq_to_prove} satisfies
\begin{align}\label{eq_exp_Vac}
\e^{i\xi V_{\rm ac}}E^U_{\rm ac}
&=\sum_{k\ge0}\frac1{k!}\;\!W_+(U,U_0,J,\Theta)\big(i\xi V_0\big)^k
W_+(U,U_0,J,\Theta)^*E^U_{\rm ac}\nonumber\\
&=W_+(U,U_0,J,\Theta)\e^{i\xi V_0}W_+(U,U_0,J,\Theta)^*E^U_{\rm ac}.
\end{align}
On another hand, if we set $W_n:=U^{-n}JU_0^n$ and $Q_0(n):=U_0^{-n}(Q\oplus Q)U_0^n$
for $n\in\Z$, then a direct calculation using the definition of the operator $J$
implies that
$$
\e^{i\xi Q(n)/n}=W_n\e^{i\xi Q_0(n)/n}W_n^*.
$$
Therefore, we obtain that
\begin{align*}
\e^{i\xi Q(n)/n}E^U_{\rm ac}-\e^{i\xi V_{\rm ac}}E^U_{\rm ac}
&=W_n\e^{i\xi Q_0(n)/n}W_n^*E^U_{\rm ac}
-W_+(U,U_0,J,\Theta)\e^{i\xi V_0}W_+(U,U_0,J,\Theta)^*E^U_{\rm ac}\\
&=I_1(n)+I_2(n)+I_3(n),
\end{align*}
with
\begin{align*}
I_1(n)&:=W_n\e^{i\xi Q_0(n)/n}\big(W_n^*-W_+(U,U_0,J,\Theta)^*\big)E^U_{\rm ac},\\
I_2(n)&:=W_n\big(\e^{i\xi Q_0(n)/n}
-\e^{i\xi V_0}\big) W_+(U,U_0,J,\Theta)^*E^U_{\rm ac},\\
I_3(n)&:=\big(W_n-W_+(U,U_0,J,\Theta)\big)
\e^{i\xi V_0}W_+(U,U_0,J,\Theta)^*E^U_{\rm ac}.
\end{align*}
But, Theorem \ref{thm_complete_walks} and the identity $E^U(\Theta)=E^U_{\rm ac}$
imply that $\slim_{n\to\infty}I_1(n)=0$. Theorem \ref{thm_complete_walks}, the
identity $E^{U_0}(\Theta)=E^{U_0}_{\rm ac}$, the inclusion
$\Ran\big(W_+(U,U_0,J,\Theta)^*\big)\subset E^{U_0}_{\rm ac}\H$, and the fact that
$\e^{i\xi V_0}E^{U_0}_{\rm ac}=E^{U_0}_{\rm ac}\e^{i\xi V_0}$ (which follows from
\eqref{eq_projectors}) imply that $\slim_{n\to\infty}I_3(n)=0$. Finally,
\cite[Thm.~4,1]{Suz16} and \cite[Thm.~VIII.20(b)]{RS1} imply that
$\slim_{n\to\infty}I_2(n)=0$. Therefore, we obtain that
$$
\slim_{n\to\infty}\big(\e^{i\xi Q(n)/n}E^U_{\rm ac}
-\e^{i\xi V_{\rm ac}}E^U_{\rm ac}\big)=0,
$$
which proves \eqref{eq_to_prove}.
\end{proof}

\begin{Remark}
If we define the asymptotic velocity operator $V\in\B(\H)$ for the full evolution
operator $U$ as
$$
V\Psi:=
\begin{cases}
0 & \mathrm{if}~\Psi\in E_{\rm p}^U\H\\
V_{\rm ac}\Psi & \mathrm{if}~\Psi\in E_{\rm ac}^U\H,
\end{cases}
$$
then the result of Proposition \ref{prop_Heisenberg} can be rephrased in the more
compact form
$$
\slim_{n\to\infty}\e^{i\xi Q(n)/n}=\e^{i\xi V},\quad\xi\in\R.
$$
\end{Remark}

We are now in a position to state the weak limit theorem. For that purpose, we denote
by $\X_n$ the random variable for the position of a quantum walker with evolution $U$
and initial normalised state $\Psi_{\rm in}\in\H$ at time $n\in\Z$. The probability
distribution of $\X_n$ is given by
$$
\P(\X_n=x):=\big\|\big(U^n\Psi_{\rm in}\big)(x)\big\|_{\C^2}^2,\quad x\in\Z,
$$
and the characteristic function of the average velocity $\X_n/n$ of the quantum walker
is given by
$$
\E(\e^{i\xi\X_n/n})
:=\big\langle U^n\Psi_{\rm in},\e^{i\xi Q/n}U^n\Psi_{\rm in}\big\rangle_\H,
=\big\langle\Psi_{\rm in},\e^{i\xi Q(n)/n}\Psi_{\rm in}\big\rangle_\H,
\quad\xi\in\R,
$$
We also use the notation $\delta_0$ for the Dirac measure on $\R$ for the point $0:$

\begin{Theorem}[Weak limit theorem]\label{thm_wlt}
Let $\Psi_{\rm in}\in\H$ with $\|\Psi_{\rm in}\|_\H=1$, let
$
\Theta:=\{\sigma(U_\ell)\cup\sigma(U_{\rm r})\}
\setminus\{\tau(U)\cup\sigma_{\rm p}(U)\}
$,
and let $\V$ be the random variable with probability distribution
\begin{equation}\label{eq_muV}
\mu_\V
:=\big\|E^U_{\rm p}\Psi_{\rm in}\big\|_\H^2\,\delta_0
+\big\|E^{V_\ell}(\;\!\cdot\;\!)
W_+(U,U_\ell,j_\ell,\Theta)^*\Psi_{\rm in}\big\|_\H^2
+\big\|E^{V_{\rm r}}(\;\!\cdot\;\!)W_+(U,U_{\rm r},j_{\rm r},\Theta)^*
\Psi_{\rm in}\big\|_\H^2
\end{equation}
with the operators $W_\pm(U,U_\star,j_\star,\Theta)$ defined in Remark
\ref{rem_sum_wave}. Then, the average velocity $\X_n/n$ converges in law to $\V$ as
$n\to\infty$, namely,
\begin{equation}\label{eq_limit}
\lim_{n\to\infty}\E(\e^{i\xi\X_n/n})=\E(\e^{i\xi\V}),	\quad\xi\in\R.
\end{equation}
\end{Theorem}

Since the average velocity $\X_n/n$ of the quantum walker converges in law to $\V$ as
$n\to\infty$, $\V$ can be interpreted as the asymptotic velocity of the quantum
walker, with $\mu_\V$ its probability distribution. Therefore, Theorem \ref{thm_wlt}
implies that the probability that the quantum walker has velocity $0$
is
$$
\P(\V=0)=\mu_\V(\{0\})\ge\big\|E^U_{\rm p}\Psi_{\rm in}\big\|_\H^2.
$$
Accordingly, we say that localisation occurs if $\P(\V=0)>0$, and \eqref{eq_muV} tells
us that this happens if the initial state $\Psi_{\rm in}$ has an overlap with the pure
point subspace of $U$. Later in this section, we will see that
$$
\P(\V=0) = \big\|E^U_{\rm p}\Psi_{\rm in}\big\|_\H^2
$$
and that localisation occurs in fact if and only if the initial state $\Psi_{\rm in}$
has an overlap with the pure point subspace of $U$.

\begin{proof}
Using Proposition \ref{prop_Heisenberg}, Equation \eqref{eq_exp_Vac},
and the identity $W_+(U,U_0,J,\Theta)^*E_{\rm ac}^U=W_+(U,U_0,J,\Theta)^*$, one obtains
\begin{align*}
\lim_{n\to\infty}\E(\e^{i\xi\X_n/n})
&=\lim_{n\to\infty}\big\langle\Psi_{\rm in},
\e^{i\xi Q(n)/n}\Psi_{\rm in}\big\rangle_\H\\
&=\big\langle\Psi_{\rm in},\big(E^U_{\rm p}
+\e^{i\xi V_{\rm ac}}E^U_{\rm ac}\big)\Psi_{\rm in}\big\rangle_\H\\
&=\big\|E_{\rm p}^U\Psi_{\rm in}\big\|_\H^2+\big\langle\Psi_{\rm in},
W_+(U,U_0,J,\Theta)\e^{i\xi V_0}W_+(U,U_0,J,\Theta)^*E^U_{\rm ac}
\Psi_{\rm in}\big\rangle_\H\\
&=\big\|E_{\rm p}^U\Psi_{\rm in}\big\|_\H^2+\big\langle W_+(U,U_0,J,\Theta)^*
E_{\rm ac}^U\Psi_{\rm in},\e^{i\xi V_0}W_+(U,U_0,J,\Theta)^*E^U_{\rm ac}
\Psi_{\rm in}\big\rangle_{\H_0}\\
&=\int_\R\e^{i\xi\lambda}\left(\big\|E_{\rm p}^U\Psi_{\rm in}\big\|_\H^2\,
\delta_0(\d\lambda)+\big\|E^{V_0}(\d\lambda)W_+(U,U_0,J,\Theta)^*E^U_{\rm ac}
\Psi_{\rm in}\big\|_{\H_0}^2\right).
\end{align*}
Therefore, to prove the claim, it only remains to observe that
\begin{align*}
& \big\|E^{V_0}(\;\!\cdot\;\!)W_+(U,U_0,J,\Theta)^*
E^U_{\rm ac}\Psi_{\rm in}\big\|_{\H_0}^2 \\
&=\big\|\big(E^{V_\ell}(\;\!\cdot\;\!)\oplus E^{V_{\rm r}}(\;\!\cdot\;\!)\big)
W_+(U,U_0,J,\Theta)^*\Psi_{\rm in}\big\|_{\H_0}^2\\
&=\big\|E^{V_\ell}(\;\!\cdot\;\!)W_+(U,U_\ell,j_\ell,\Theta)^*\Psi_{\rm in}\big\|_\H^2
+\big\|E^{V_{\rm r}}(\;\!\cdot\;\!)W_+(U,U_{\rm r},j_{\rm r},\Theta)^*
\Psi_{\rm in}\big\|_\H^2.
\end{align*}
\end{proof}

In order to give a better description of the probability distribution $\mu_\V$, we
recall from Section \ref{sec_walks} that we have for $k\in[0,2\pi)$
\begin{align*}
\widehat{U_\star}(k)=\sum_{j=1}^2\lambda_{\star,j}(k)\Pi_{\star,j}(k),
\quad\widehat{V_\star}(k)=\sum_{j=1}^2v_{\star,j}(k)\Pi_{\star,j}(k)
\quad\hbox{and}\quad
v_{\star,j}(k)=i\;\!\lambda_{\star,j}'(k)\big(\lambda_{\star,j}(k)\big)^{-1}.
\end{align*}
We also recall from \cite[Lemma~4.2]{RST_1} the following properties of the functions
$v_{\star,j}$ given in terms of the parameters $a_\star,\alpha_\star,\delta_\star$ of
\eqref{eq_param}:
\begin{enumerate}
\item[(i)] If $a_\star=0$, then $v_{\star,j}=0$ for $j\in\{1,2\}$.
\item[(ii)] If $a_\star\in(0,1)$, then
$v_{\star,j}(k)=\frac{(-1)^{j}\varsigma_\star(k)}{\eta_\star(k)}$ for $j\in\{1,2\}$,
$k\in[0,2\pi)$ and
$$
\tau_\star(k):=a_\star\cos(k+\alpha_\star-\delta_\star/2),\quad
\eta_\star(k):=\sqrt{1-\tau_\star(k)^2},\quad
\varsigma_\star(k):=a_\star\sin(k+\alpha_\star-\delta_\star/2).
$$
\item[(iii)] If $a_\star=1$, then $v_{\star,j}(k)=(-1)^{j}$ for $j\in\{1,2\}$ and
$k\in[0,2\pi)$.
\end{enumerate}

With this done, we can start our study of the probability distribution $\mu_\V$ by
collecting some information on the operators $W_+(U,U_\star,j_\star,\Theta)$ appearing
in \eqref{eq_muV}:

\begin{Lemma}\label{Lemma_W}
Let
$
\Theta:=\{\sigma(U_\ell)\cup\sigma(U_{\rm r})\}
\setminus\{\tau(U)\cup\sigma_{\rm p}(U)\}
$.
\begin{enumerate}
\item[(a)] If $a_\star=0$, then $W_+(U,U_\star,j_\star,\Theta)=0$.
\item[(b)] If  $a_\ell\in(0,1)$, then $W_+(U,U_\ell,j_\ell,\Theta)$ is a partial
isometry with initial subspace $\chi_-(V_\ell)\H$, and if  $a_{\rm r}\in(0,1)$, then
$W_+(U,U_{\rm r},j_{\rm r},\Theta)$ is a partial isometry with initial subspace
$\chi_+(V_{\rm r})\H$.
\item[(c)] If $a_\ell=1$, then $W_+(U,U_\ell,j_\ell,\Theta)$ is a partial isometry
with initial subspace
$\ell^2\left(\Z,\left(\begin{smallmatrix}\C \\ 0\end{smallmatrix}\right)\right)$, and
if $a_{\rm r}=1$, then $W_+(U,U_{\rm r},j_{\rm r},\Theta)$ is a partial isometry with
initial subspace
$\ell^2\left(\Z,\left(\begin{smallmatrix} 0 \\ \C\end{smallmatrix}\right)\right)$.
\end{enumerate}
\end{Lemma}

\begin{proof}
The claim (a) is a direct consequence of Remark \ref{rem_sum_wave} and Equation
\eqref{eq_projectors}. The claim (b) is a direct consequence of Remark
\ref{rem_sum_wave}, Equation \eqref{eq_projectors} and Proposition \ref{prop_kernels}.
For the claim (c), we recall from Proposition \ref{prop_kernels} and the point (iii)
above that in the case $a_\ell=1$ the initial subspace of $W_+(U,U_\ell,j_\ell,\Theta)$
coincides with the eigenspace of $V_\ell$ for the eigenvalue $-1$ (see the case $j=1$
in the point (iii) above). Then, the formula for $u_{\star,1}(k)$
in \cite[Sec.~4.1]{RST_1} directly implies that this eigenspace is equal to the space
$\ell^2\left(\Z,\left(\begin{smallmatrix}\C \\ 0\end{smallmatrix}\right)\right)$.
The proof of the claim in the case $a_{\rm r}=1$ is similar.
\end{proof}

Now, to pursue our study of the probability distribution $\mu_\V$, we recall the
definition of the Konno function
$$
f_{\rm K}:\R\times(0,1]\to[0,\infty),~~(\upsilon,r)\mapsto
\begin{cases}
\frac{\sqrt{1-r^2}}{\pi(1-\upsilon^2)\sqrt{r^2-\upsilon^2}}
& \hbox{if $|\upsilon|<r$,}\\ 0 & \hbox{otherwise.}
\end{cases}
$$
With this definition at hand, we can establish the following:

\begin{Proposition}\label{prop_meas}
Let $\Psi_{\rm in}\in\H$ with $\|\Psi_{\rm in}\|_\H=1$, and let
$
\Theta
:=\{\sigma(U_\ell)\cup\sigma(U_{\rm r})\}\setminus\{\tau(U)\cup\sigma_{\rm p}(U)\}
$.
\begin{enumerate}
\item[(a)] If $a_\star=0$, then
$
\big\|E^{V_\star}(\;\!\cdot\;\!)W_+(U,U_\star,j_\star,\Theta)^*
\Psi_{\rm in}\big\|_\H^2=0
$.
\item[(b)] If $a_\star\in(0,1)$ and $\upsilon\in\R$, then
$$
\frac{\d\big\|E^{V_\star}\big((-\infty,\upsilon]\big)W_+(U,U_\star,j_\star,\Theta)^*
\Psi_{\rm in}\big\|_\H^2}{\d\upsilon}
=\begin{cases}
\frac12\chi_{[-a_\ell,0)}(\upsilon)\;\!\omega_\ell(\upsilon)
f_{\rm K}(\upsilon,a_\ell) & \mathrm{if}~\star=\ell\smallskip\\
\frac12\chi_{(0,a_{\rm r}]}(\upsilon)\;\!\omega_{\rm r}(\upsilon)
f_{\rm K}(\upsilon,a_{\rm r}) & \mathrm{if}~\star={\rm r}
\end{cases}
$$
for some nonnegative functions
$
\omega_\ell
\in\lone\big([-a_\ell,0),\tfrac12f_{\rm K}(\;\!\cdot\;\!,a_\ell)\d\upsilon\big)
$
and
$
\omega_{\rm r}
\in\lone\big((0,a_{\rm r}],\tfrac12f_{\rm K}(\;\!\cdot\;\!,a_{\rm r})\d\upsilon\big)
$.
\item[(c)]  If $a_\star=1$, then
$$
\big\|E^{V_\star}(\;\!\cdot\;\!)W_+(U,U_\star,j_\star,\Theta)^*\Psi_{\rm in}\big\|_\H^2
=\begin{cases}
\big\|W_+(U,U_\ell,j_\ell,\Theta)^*\Psi_{\rm in}\big\|_\H^2\,\delta_{-1}
& \mathrm{if}~\star=\ell\smallskip\\
\big\|W_+(U,U_{\rm r},j_{\rm r},\Theta)^*\Psi_{\rm in}\big\|_\H^2\,\delta_1
& \mathrm{if}~\star={\rm r}
\end{cases}
$$
with $\delta_{\pm1}$ the Dirac measures on $\R$ for the points $\pm1$.
\end{enumerate}
\end{Proposition}

\begin{proof}[Proof of (a) and (c)]
In the case $a_\ell=0$, the claim (a) follows directly from Lemma \ref{Lemma_W}(a).
In the case of $a_\ell=1$, we know from Lemma \ref{Lemma_W}(c) that
$W_+(U,U_\ell,j_\ell,\Theta)$ has initial subspace
$\ell^2\left(\Z,\left(\begin{smallmatrix}\C \\ 0\end{smallmatrix}\right)\right)$,
so that
$
W_+(U,U_\ell,j_\ell,\Theta)^*\Psi_{\rm in}
\in\ell^2\left(\Z,\left(\begin{smallmatrix}\C \\ 0\end{smallmatrix}\right)\right)
$.
On the other hand, we have shown in \cite[Sec.~4.1]{RST_1} that
$\sigma(V_\ell)=\sigma_{\rm p}(V_\ell)=\{-1,1\}$, with
$\ell^2\left(\Z,\left(\begin{smallmatrix}\C \\ 0\end{smallmatrix}\right)\right)$
the eigenspace associated to the eigenvalue $-1$. Thus, we obtain for any Borel set
$B\subset\R$ that
$$
\big\|E^{V_\ell}(B)W_+(U,U_\ell,j_\ell,\Theta)^*\Psi_{\rm in}\big\|_\H^2
=\begin{cases}
\big\|W_+(U,U_\ell,j_\ell,\Theta)^*\Psi_{\rm in}\big\|_\H^2 & \hbox{if $-1\in B$,}\\
0 & \hbox{otherwise.}
\end{cases}
$$
Since a similar result holds for the operator $W_+(U,U_{\rm r},j_{\rm r},\Theta)$,
the eigenspace
$\ell^2\left(\Z,\left(\begin{smallmatrix} 0 \\ \C\end{smallmatrix}\right)\right)$ and
the eigenvalue $+1$, we infer the result of the claim (c).
\end{proof}

In order to prove the claim (b) of Proposition \ref{prop_meas}, we need two
preparatory lemmas. The first one is the following:

\begin{Lemma}
Define for $\star\in\{\ell,{\rm r}\}$, $j\in\{1,2\}$, $m\in\{0,1\}$ and $a_\star\in(0,1)$
the function
$$
k_{\star,j,m }:[-a_\star,a_\star]\to I_m,~~
\upsilon\mapsto\frac{\delta_\star}2-\alpha_\star+m\;\!\pi
+\arcsin\left(\frac{(-1)^{j+m}b_\star\upsilon}{a_\star\sqrt{1-\upsilon^2}}\right),
$$
where
$$
I_0:=\big[-\pi/2-\alpha_\star+\delta_\star/2,\pi/2-\alpha_\star+\delta_\star/2\big]
\quad\hbox{and}\quad
I_1:=I_0+\pi.
$$
Then, the function $k_{\star,j,m}$ is differentiable on $(-a_\star,a_\star)$ with
\begin{equation}\label{eq_k_diff}
k_{\star,j,m}'(\upsilon)=(-1)^{j+m}\pi\;\!f_{\rm K}(\upsilon,a_\star),
\quad\upsilon\in(-a_\star,a_\star),
\end{equation}
and $k_{\star,j,m}$ is a bijection with inverse $v_{\star,j}|_{I_m}$, that is,
\begin{equation}\label{eq_inverses}
\big(v_{\star,j}\circ k_{\star,j,m}\big)(\upsilon)=\upsilon
~\mathrm{for}~\upsilon\in[-a_\star,a_\star]\quad\mathrm{and}\quad
\big(k_{\star,j,m}\circ v_{\star,j}\big)(k)=k~\mathrm{for}~k\in I_m.
\end{equation}
\end{Lemma}

\begin{proof}
Set $f_\star(\upsilon):=\frac{b_\star\upsilon}{a_\star\sqrt{1-\upsilon^2}}$ for
$\upsilon\in[-a_\star,a_\star]$. Since
$f_\star'(\upsilon)=\frac{b_\star}{a_\star(1-\upsilon^2)^{3/2}}>0$ on
$(-a_\star,a_\star)$, the function $f_\star$ is strictly increasing on
$(-a_\star,a_\star)$ with $f_\star([-a_\star,a_\star])=[-1,1]$. It follows that
$k_{\star,j,m}$ is differentiable on $(-a_\star, a_\star)$ with derivative
$$
k_{\star,j,m}'(\upsilon)
=\frac{\d\big(\arcsin\circ(-1)^{j+m}f_\star\big)}{\d\upsilon}(\upsilon)
=\frac{(-1)^{j+m}f_\star'(\upsilon)}{\sqrt{1-f_\star(\upsilon)^2}}
=\frac{(-1)^{j+m}\sqrt{1-a_\star^2}}{\sqrt{a_\star^2-\upsilon^2}\;\!(1-\upsilon^2)}.
$$
This implies \eqref{eq_k_diff}.

Now, the fact that $k_{\star,j,m}'\ne0$ on $(-a_\star,a_\star)$ implies that
$k_{\star,j,m}$ is invertible. Since
$\varsigma_\star\big(k_{\star,j,m}(\upsilon)\big)=(-1)^ja_\star f_\star(\upsilon)$ and
$
\eta_\star\big(k_{\star,j,m}(\upsilon)\big)
=\sqrt{b_\star^2+a_\star^2f_\star^2(\upsilon)}
$
for $\upsilon\in[-a_\star,a_\star]$, one also obtains that
$$
v_{\star,j}\big(k_{\star,j,m}(\upsilon)\big)
=\frac{(-1)^j\varsigma_\star\big(k_{\star,j,m}(\upsilon)\big)}
{\eta_\star\big(k_{\star,j,m}(\upsilon)\big)}
=\frac{a_\star f_\star(\upsilon)}{\sqrt{b_\star^2+a_\star^2f_\star(\upsilon)^2}}
=\upsilon.
$$
On the other hand, since
$f_\star\big(v_{\star,j}(k)\big)=(-1)^j\sin(k+\alpha_\star-\delta_\star/2)$ for
$k\in[0,2\pi)$, one obtains that
$$
k_{\star,j,m}\big(v_{\star,j}(k)\big)
=\frac{\delta_\star}2-\alpha_\star+m\pi
+\arcsin\big((-1)^m\sin(k+\alpha_\star-\delta_\star/2)\big).
$$
Therefore, if $m=0$ and $k\in I_0$, then
$k+\alpha_\star-\delta_\star/2\in[-\pi/2,\pi/2]$ and thus
$k_{\star,j,m}(v_{\star,j}(k))=k$. And if $m=1$ and $k\in I_1$, then
$k+\alpha_\star-\delta_\star/2\in[\pi/2,3\pi/2]$ and
$(-1)^m\sin(k+\alpha_\star-\delta_\star/2)=\sin(k+\alpha_\star-\delta_\star/2-\pi)$.
Thus $k+\alpha_\star-\delta_\star/2-\pi\in[-\pi/2,\pi/2]$, and one obtains once again
$k_{\star,j,m}\big(v_{\star,j}(k)\big)=k$, which concludes the proof.
\end{proof}

To state our second preparatory lemma, we need to introduce some notations. For
$k\in[0,2\pi)$, we denote by $u_{\star,j}(k)\in\C^2$ a normalised eigenvector of the
operator $\widehat{U_\star}(k)$ for the eigenvalue $\lambda_{\star,j}(k)$.
In \cite[Sec.~4]{RST_1}, we have shown that $u_{\star,j}(k)$ can be chosen
$C^\infty$ in the variable $k$. For $\star\in\{\ell,{\rm r}\}$, we set
$$
\G_\star
:=\ltwo\big([-a_\star,a_\star],\tfrac12f_{\rm K}(\;\!\cdot\;\!,a_\star)\d\upsilon\big).
$$
Finally, for $j\in\{1,2\}$ and $m\in\{0,1\}$, we define the operator
$K_{\star,j,m}:\H\to\G_\star$ by
$$
\big(K_{\star,j,m}\Psi\big)(\upsilon)
:= \big\langle u_{\star,j}\big(k_{\star,j,m}(\upsilon)\big),
(\F\Psi)\big(k_{\star,j,m}(\upsilon)\big) \big\rangle_{\C^2},
\quad\Psi\in\H,~\upsilon\in[-a_\star,a_\star].
$$

\begin{Lemma}\label{Lemma_K}
The operator $K_{\star,j,m}$ is bounded and satisfies the following:
\begin{enumerate}
\item[(a)] For $g\in\G_\star$, one has
$
(K_{\star,j,m})^*g
=\F^*\big\{\chi_{I_m}g\big(v_{\star,j}(\;\!\cdot\;\!)\big)u_{\star,j}\big\}
$.
\item[(b)] For $j,j'\in\{1,2\}$ and $m,m'\in\{0,1\}$, one has
$K_{\star,j,m}(K_{\star,j',m'})^*=\delta_{j,j'}\delta_{m,m'}\;\!{\rm id}_{\G_\star}$.
\item[(c)]
$\sum_{j\in\{1,2\},m\in\{0,1\}}(K_{\star,j,m})^*K_{\star,j,m}={\rm id}_\H$.
\item[(d)] For any Borel function $F:[-a_\star, a_\star]\to\C$, the multiplication
operator by $F$ in $\G_\star$ (denoted by the same symbol) satisfies the equation
$$
\sum_{j\in\{1,2\},m\in\{0,1\}}(K_{\star,j,m})^*F K_{\star,j,m}=F(V_\star).
$$
\end{enumerate}
\end{Lemma}

\begin{proof}
Using the change of variables $k:=k_{\star,j,m}(\upsilon)$, Equation \eqref{eq_k_diff},
and the normalisation of $u_{\star,j}(k)$, we obtain for any $\Psi\in\H$
\begin{align*}
\big\|K_{\star,j,m}\Psi\big\|_{\G_\star}^2
&=\tfrac12\int_{-a_\star}^{a_\star}\big|\big(K_{\star,j,m}\Psi\big)(\upsilon)\big|^2
f_{\rm K}(\upsilon,a_\star)\;\!\d\upsilon\\
&=\tfrac1{2\pi}\int_{I_m}\big|\big\langle u_{\star,j}(k),
(\F\Psi)(k)\big\rangle_{\C^2}\big|^2\;\!\d k\\
&\le\tfrac1{2\pi}\int_0^{2\pi}\big|(\F\Psi)(k)\big|_{\C^2}^2\;\!\d k\\
&=\|\Psi\|_\H^2,
\end{align*}
which proves that $K_{\star,j,m}$ is bounded. Furthermore, using
the change of variables $k:=k_{\star,j,m}(\upsilon)$ and \eqref{eq_inverses},
we obtain for any $g\in\G_\star$
\begin{align*}
\big\langle g,K_{\star,j,m}\Psi\big\rangle_{\G_\star}
&=\tfrac12\int_{-a_\star}^{a_\star}\overline g(\upsilon)
\big\langle u_{\star,j}\big(k_{\star,j,m}(\upsilon)\big),
(\F\Psi)\big(k_{\star,j,m}(\upsilon)\big)\big\rangle_{\C^2}
f_{\rm K}(\upsilon,a_\star)\;\!\d\upsilon \\
&=\tfrac1{2\pi}\int_{I_m}\overline g\big(v_{\star,j}(k)\big)
\big\langle u_{\star,j}(k),(\F\Psi)(k)\big\rangle_{\C^2}\;\!\d k\\
&=\big\langle\F^*\big\{\chi_{I_m}g\big(v_{\star,j}(\;\!\cdot\;\!)\big)
u_{\star,j}\big\},\Psi\big\rangle_\H,
\end{align*}
which proves (a). To prove (b), we observe that points (a) and \eqref{eq_inverses} imply
for any $g\in\G_\star$ and $\upsilon\in[-a_\star,a_\star]$	
\begin{align*}
\big(K_{\star,j,m}(K_{\star,j',m'})^*g\big)(\upsilon)
&=\big\langle u_{\star,j}\big(k_{\star,j,m}(\upsilon)\big),
\big(\F(K_{\star,j',m'})^*g\big)\big(k_{\star,j,m}(\upsilon)\big)\big\rangle_{\C^2}\\
&=\chi_{I_{m'}}\big((k_{\star,j,m}(\upsilon)\big)
g\big(v_{\star,j'}\big(k_{\star,j,m}(\upsilon)\big)\big)
\big\langle u_{\star,j}\big(k_{\star,j,m}(\upsilon)\big),
u_{\star,j'}\big(k_{\star,j,m}(\upsilon)\big)\big\rangle_{\C^2}\\
&=\delta_{j,j'}\delta_{m,m'}g(\upsilon).
\end{align*}
On the other hand, for any $\Psi\in\H$ and $k\in[0,2\pi)$ we have
\begin{align*}
\big(\F(K_{\star,j,m})^*K_{\star,j,m}\Psi\big)(k)
&=\chi_{I_m}(k)\big(K_{\star,j,m}\Psi\big)\big(v_{\star,j}(k)\big)u_{\star,j}(k)\\
&=\chi_{I_m}(k)\big\langle u_{\star,j}\big(k_{\star,j,m}\big(v_{\star,j}(k)\big)\big),
(\F\Psi)\big(k_{\star,j,m}\big(v_{\star,j}(k)\big)\big)
\big\rangle_{\C^2}u_{\star,j}(k)\\
&=\chi_{I_m}(k)\big\langle u_{\star,j}(k),(\F\Psi)(k)\big\rangle_{\C^2}u_{\star,j}(k).
\end{align*}
Hence, we obtain
$$
(K_{\star,j,m})^* K_{\star,j,m}\Psi
=\F^*\big\{\chi_{I_m}\big\langle u_{\star,j}(\;\!\cdot\;\!),
(\F\Psi)(\;\!\cdot\;\!)\big\rangle_{\C^2}u_{\star,j}\big\}
=\F^*\chi_{I_m}\Pi_{\star,j}(\;\!\cdot\;\!)\F\Psi,
$$
and thus
$$
\sum_{j\in\{1,2\},m\in\{0,1\}}(K_{\star,j,m})^*K_{\star,j,m}\Psi
=\sum_{j\in\{1,2\},m\in\{0,1\}}\F^*\chi_{I_m}\Pi_{\star,j}(\;\!\cdot\;\!)\F\Psi
=\Psi,
$$
which shows (c). Finally, for any $\Psi\in\H$ and $k\in[0,2\pi)$ we have
\begin{align*}
\big(\F(K_{\star,j,m})^*FK_{\star,j,m}\Psi\big)(k)
&=\chi_{I_m}(k)\big(F K_{\star,j,m}\Psi\big)\big(v_{\star,j}(k)\big)u_{\star,j}(k)\\
&=\chi_{I_m}(k)F\big(v_{\star,j}(k)\big)\Pi_{\star,j}(k)\F\Psi
\end{align*}
which implies that
$$
\sum_{j\in\{1,2\},m\in\{0,1\}}(K_{\star,j,m})^*FK_{\star,j,m}
=\sum_{j\in\{1,2\}}\F^*F\big(v_{\star,j}(\;\!\cdot\;\!)\big)
\Pi_{\star,j}(\;\!\cdot\;\!)\F
=\F^*F(\widehat{V_\star})\F
=F(V_\star),
$$
as stated in (d).
\end{proof}

We can now provide the proof of the claim (b) of Proposition \ref{prop_meas}:

\begin{proof}[Proof of Proposition \ref{prop_meas}(b)]
Take a Borel set $B\subset\R$. Then, Lemma \ref{Lemma_W}(b), Lemma \ref{Lemma_K}(d)
and \cite[Lemma~4.2(b)]{RST_1} imply that
\begin{align*}
E^{V_\ell}(B)W_+(U,U_\ell,j_\ell,\Theta)^*\Psi_{\rm in}
&=\chi_B(V_\ell)\chi_-(V_\ell)W_+(U,U_\ell,j_\ell,\Theta)^*\Psi_{\rm in}\\
&=\sum_{j\in\{1,2\},m\in\{0,1\}}(K_{\ell,j,m})^*\chi_{B\cap[-a_\ell,0)}K_{\ell,j,m}
W_+(U,U_\ell,j_\ell,\Theta)^*\Psi_{\rm in}.
\end{align*}
Thus, it follows by Lemma \ref{Lemma_K}(b) that
\begin{align*}
&\big\|E^{V_\ell}(B)W_+(U,U_\ell,j_\ell,\Theta)^*\Psi_{\rm in}\big\|^2_\H\\
&=\sum_{j\in\{1,2\},m\in\{0,1\}}\big\|\chi_{B \cap[-a_\ell,0)}K_{\ell,j,m}
W_+(U,U_\ell,j_\ell,\Theta)^*\Psi_{\rm in}\big\|_{\G_\ell}^2\\
&=\tfrac12\sum_{j\in\{1,2\},m\in\{0,1\}}\int_{B\cap[-a_\ell,0)}
\big|\big(K_{\ell,j,m}W_+(U,U_\ell,j_\ell,\Theta)^*\Psi_{\rm in}\big)(\upsilon)\big|^2
f_{\rm K}(\upsilon,a_\ell)\;\!\d\upsilon,
\end{align*}
which means that the density function of
$\big\|E^{V_\ell}(B)W_+(U,U_\ell,j_\ell,\Theta)^*\Psi_{\rm in}\big\|^2_\H$ is given by
\begin{align*}
&\frac{\d\big\|E^{V_\ell}\big((-\infty,\upsilon]\big)W_+(U,U_\ell,j_\ell,\Theta)^*
\Psi_{\rm in}\big\|_\H^2}{\d\upsilon}\\
&=\tfrac12\sum_{j\in\{1,2\},m\in\{0,1\}}\chi_{[-a_\ell,0)}(\upsilon)
\big|\big(K_{\ell,j,m}W_+(U,U_\ell,j_\ell,\Theta)^*\Psi_{\rm in}\big)(\upsilon)\big|^2
f_{\rm K}(\upsilon,a_\ell).
\end{align*}
Since a similar argument leads to
\begin{align*}
&\frac{\d\big\|E^{V_{\rm r}}\big((-\infty,\upsilon]\big)
W_+(U,U_{\rm r},j_{\rm r},\Theta)^*\Psi_{\rm in}\big\|_\H^2}{\d\upsilon}\\
&=\tfrac12\sum_{j\in\{1,2\},m\in\{0,1\}}\chi_{(0,a_{\rm r}]}(\upsilon)\big|
\big(K_{{\rm r},j,m}W_+(U,U_{\rm r},j_{\rm r},\Theta)^*
\Psi_{\rm in}\big)(\upsilon)\big|^2f_{\rm K}(\upsilon,a_{\rm r}),
\end{align*}
the claim follows by setting
\begin{equation}\label{last_label}
\omega_\star(\upsilon)
:=\sum_{j\in\{1,2\},m\in\{0,1\}}\big|(K_{\star,j,m}W_+(U,U_\star,j_\star,\Theta)^*
\Psi_{\rm in})(\upsilon)\big|^2.
\end{equation}
\end{proof}

In our final theorem, we gather the results obtained so far on the probability
distribution $\mu_\V$ and we prove a uniqueness result.

\begin{Theorem}\label{Thm_wlt_meas}
Let $\Psi_{\rm in}\in\H$ with $\|\Psi_{\rm in}\|_\H=1$, and let $\V$ be the random
variable defined by \eqref{eq_limit}. Then, $\V$ has probability distribution
\begin{align}
\mu_\V(\d\upsilon)
&=\kappa_0\;\!\delta_0(\d\upsilon)+\kappa_\ell\;\!\delta_{-1}(\d\upsilon)
+\kappa_{\rm r}\;\!\delta_1(\d\upsilon)\nonumber\\
&\quad+\chi_{[-a_\ell,0)}(\upsilon)w_\ell(\upsilon)\tfrac12f_{\rm K}(\upsilon,a_\ell)
\;\!\d\upsilon+\chi_{(0,a_{\rm r}]}(\upsilon)w_{\rm r}(\upsilon)\tfrac12
f_{\rm K}(\upsilon,a_{\rm r})\;\!\d\upsilon,\quad\upsilon\in\R,
\label{eq_muV_final}
\end{align}
with $\kappa_0,\kappa_\star\geq 0$ given by
$$
\kappa_0:=\big\|E^U_{\rm p}\Psi_{\rm in}\big\|_\H^2
\quad\hbox{and}\quad
\kappa_\star
:=\begin{cases}
\big\|W_+(U,U_\star,j_\star,\Theta)^*\Psi_{\rm in}\big\|_\H^2&\mathrm{if}~a_\star=1\\
0&\mathrm{otherwise,}
\end{cases}
$$
and with
$
w_\ell
\in\lone\big([-a_\ell,0),\tfrac12f_{\rm K}(\;\!\cdot\;\!,a_\ell)\;\!\d\upsilon\big)
$
and
$
w_{\rm r}
\in\lone\big((0,a_{\rm r}],\tfrac12f_{\rm K}(\;\!\cdot\;\!,a_{\rm r})\;\!\d\upsilon\big)
$
given by
$$
w_\star(\upsilon)
:=\begin{cases}
\sum_{j\in\{1,2\},m\in\{0,1\}}\big|\big(K_{\star,j,m}W_+(U,U_\star,j_\star,\Theta)^*
\Psi_{\rm in}\big)(\upsilon)\big|^2 & \mathrm{if}~a_\star\in(0,1)\\
0 & \mathrm{otherwise.}
\end{cases}
$$
Furthermore, the decomposition of $\mu_\V$ is unique.
\end{Theorem}

\begin{proof}
In view of Theorem \ref{thm_wlt} and Proposition \ref{prop_meas}, only the uniqueness
of the decomposition of $\mu_\V$ has to be established. For that purpose, we observe
that $\mu_\V$ is the sum of a pure point measure located at the three distinct points
$-1,0,1$ and an absolutely continuous measure. As a consequence, the coefficients
$\kappa_\ell,\kappa_0,\kappa_{\rm r}$ are unique. In addition, since the absolutely
continuous measure is the sum of two absolutely continuous measures with disjoint
supports $[-a_\ell,0)$ and $(0,a_{\rm r}]$, each of these measures is unique in the
$\lone$-sense given by their density functions in
$\lone\big([-a_\ell,0),\tfrac12f_{\rm K}(\;\!\cdot\;\!,a_\ell)\;\!\d\upsilon\big)$ and
$\lone\big((0,a_{\rm r}],\tfrac12f_{\rm K}(\;\!\cdot\;\!,a_{\rm r})\;\!\d\upsilon\big)$.
\end{proof}

\begin{Remark}\label{rem_wlt}
In the case of the one-dimensional random walk where the walker moves to the left with
probability $p$ and to the right with probability $q$, the classical central limit
theorem implies that the random variable $\frac{\X_n-n(q-p)}{2\sqrt{npq}}$ converges
in law as $n\to\infty$ to a random variable $\mathtt Z$ with standard normal
distribution ${\rm N}(0,1)$. Therefore, the weak limit theorem \ref{thm_wlt} can be
interpreted as a quantum analogue of the classical central limit theorem. On the other
hand, the classical central limit theorem implies that the average velocity $\X_n/n$
always converges in law as $n\to\infty$ to a Gaussian random variable with
distribution ${\rm N}(q-p, 4pq/n)$, whereas the weak limit theorem leads to a variety
of limit distributions $\mu_\V$ depending on the outgoing states
$W_+(U,U_\star,j_\star,\Theta)^*\Psi_{\rm in}$.
\end{Remark}

To conclude, we discuss in more detail some particular cases of Theorem
\ref{Thm_wlt_meas}:

\begin{Example}\label{ex_wlt}
(a) In the case $a_\ell=0$ and $a_{\rm r}\in(0,1)$, the formula \eqref{eq_muV_final}
reduces to
$$
\mu_\V(\d\upsilon)
=\kappa_0\;\!\delta_0(\d\upsilon)+\tfrac12\chi_{(0,a_{\rm r}]}(\upsilon)
w_{\rm r}(\upsilon)f_{\rm K}(\upsilon,a_{\rm r})\;\!\d\upsilon.
$$
Thus, the density function of the absolutely continuous part of $\mu_\V$ is supported
in $(0,a_{\rm r}]$, and the quantum walker can asymptotically move only to the right
at a speed belonging to the interval $(0,a_{\rm r}]$.

(b) In the case $a_\ell=1$ and $a_{\rm r}\in(0,1)$, the formula \eqref{eq_muV_final}
reduces to
$$
\mu_\V(\d\upsilon)
=\kappa_0\;\!\delta_0(\d\upsilon)+\kappa_\ell\;\!\delta_{-1}(\d\upsilon)
+\tfrac12\chi_{(0,a_{\rm r}]}(\upsilon) w_{\rm r}(\upsilon)
f_{\rm K}(\upsilon,a_{\rm r})\;\!\d\upsilon.
$$
Thus, the quantum walker can asymptotically move to the left at speed $1$ or to the
right at a speed belonging to the interval $(0,a_{\rm r}]$.

(c) In the case $a_\ell, a_{\rm r}\in(0,1)$, the formula \eqref{eq_muV_final} reduces
to
\begin{equation}\label{muV_iso}
\mu_\V(\d\upsilon)
=\kappa_0\;\!\delta_0(\d\upsilon)+\tfrac12\chi_{[-a_\ell,0)}(\upsilon)w_\ell(\upsilon)
f_{\rm K}(\upsilon,a_\ell)\;\!\d\upsilon
+\tfrac12\chi_{(0,a_{\rm r}]}(\upsilon)w_{\rm r}(\upsilon)
f_{\rm K}(\upsilon,a_{\rm r})\;\!\d\upsilon.
\end{equation}
Thus, the quantum walker can asymptotically move to the left at a speed belonging to
the interval $(0,a_\ell]$ or to the right at a speed belonging to the interval
$(0,a_{\rm r}]$. In particular, if $a_\ell>a_{\rm r}$, then the quantum walker can
move faster on the left-hand side than on the right-hand side.
\end{Example}

Finally, we note that the formula \eqref{muV_iso} covers various previously known
results:

\begin{Example}\label{ex_cpr}
(a) In the case $a_\ell=a_{\rm r}=1/\sqrt2$, the formula \eqref{muV_iso} reduces to
\begin{equation}\label{eq_nous}
\mu_\V(\d\upsilon)
=\kappa_0\;\!\delta_0(\d\upsilon)
+\tfrac12\big(\chi_{[-1/\sqrt2,0)}(\upsilon)w_\ell(\upsilon)
+\chi_{(0,1/\sqrt2]}(\upsilon)w_{\rm r}(\upsilon)\big)
f_{\rm K}(\upsilon,1/\sqrt2)\;\!\d\upsilon.
\end{equation}
This generalises the formula obtained by Endo et al. in the case of a two-phase
quantum walk with one defect \cite[Thm.~2.1]{EEKST16}. Indeed, in their work Endo et
al. consider a coin operator given by
$$
C(x)=
\begin{cases}
C_+=\frac1{\sqrt2}
\big(\begin{smallmatrix}1&\e^{i\sigma_+}\\
\e^{-i\sigma_+}&-1\end{smallmatrix}\big)
& \mathrm{if}~x\ge1\medskip\\
C_0=\big(\begin{smallmatrix}1&0\\0&-1\end{smallmatrix}\big)
& \mathrm{if}~x=0\medskip\\
C_-=\frac1{\sqrt2}\big(\begin{smallmatrix}1&\e^{i\sigma_-}
\\ \e^{-i\sigma_-}&-1\end{smallmatrix}\big)
& \mathrm{if}~x\le-1,
\end{cases}
$$
for some $\sigma_\pm\in[0,2\pi)$, and they obtain the formula
\begin{equation}\label{eq_eux}
\mu_\V(\d\upsilon)
=c\;\!\delta_0(\d\upsilon)+w(\upsilon)f_{\rm K}(\upsilon,1/\sqrt2)\;\!\d\upsilon
\end{equation}
where $c:=\sum_{x\in\Z}\big(\lim_{N\to\infty}\frac1N\sum_{n=0}^{N-1}\P(\X_n=x)\big)$
and $w:\R\to\R$ is some particular function. However, by applying a discrete analogue
of the RAGE theorem, or Wiener's theorem, or in a way similar to
\cite[Appendix]{SS16}, one can prove that
$c=\kappa_0=\big\|E^U_{\rm p}\Psi_{\rm in}\big\|_\H^2$, thus showing that
\eqref{eq_nous} is a generalisation of \eqref{eq_eux}. Moreover, the uniqueness of the
decomposition \eqref{muV_iso} leads to the explicit formula for $w$:
$$
w(\upsilon)
=\tfrac12\big(\chi_{[-1/\sqrt2,0)}(\upsilon)w_\ell(\upsilon)
+\chi_{(0,1/\sqrt2]}(\upsilon)w_{\rm r}(\upsilon)\big),
$$
Endo et al. obtained an explicit expression for $w$ only in the case the initial state
satisfies $\Psi_{\rm in}(x)=0$ everywhere except for $x=0$.

(b) The isotropic case where $C(x)=C_\infty+O(|x|^{-1-\varepsilon})$,
$C_\infty\in\U(2)$, and $a_\infty:=|(C_\infty)_{1,1}|\in(0,1)$, was studied in
\cite{Suz16} by the second author. In our setup, this corresponds to setting
$C_\ell=C_{\rm r}=C_\infty$ and $\varepsilon_\ell=\varepsilon_{\rm r}=\varepsilon$ in
Assumption \ref{ass_short} and having $a_\ell=a_{\rm r}=a_\infty\in (0,1)$ in
\eqref{eq_param}. In \cite{Suz16}, it was shown that
$$
\mu_\V(\d\upsilon)
=\kappa_0\;\!\delta_0(\d\upsilon)+w(\upsilon)f_{\rm K}(\upsilon,a_\infty)\;\!\d\upsilon,
$$
where
$
w\in\lone\big([-a_\infty,a_\infty],
\tfrac12f_{\rm K}(\;\!\cdot\;\!,a_\infty)\;\!\d\upsilon\big)
$
is a function similar to \eqref{last_label} but defined in terms of the wave operator
$W_+(U,U_\infty,1_\H,\Theta)$, with the identity $1_\H$ instead of the identification
operator $J$ and the evolution operator $U_\infty:=SC_\infty$ in $\H$ instead of the
evolution operator $U_0$ in $\H_0$. Thus, by using once again the uniqueness of the
decomposition of \eqref{muV_iso}, we obtain the explicit formula for $w$:
$$
w(\upsilon)
=\tfrac12\chi_{[-a_\infty,0)}(\upsilon)w_\ell(\upsilon)
+\tfrac12\chi_{(0,a_\infty]}(\upsilon)w_{\rm r}(\upsilon).
$$
In the one-defect case, where $C(x)=C_\infty$ everywhere except for $x=0$, Konno et
al. obtained in \cite{Kon10} an explicit expression of $w$ only in the case the initial
state satisfies $\Psi_{\rm in}(x)=0$ everywhere except for $x=0$.

(c) The homogeneous case, where $C(x)=C_\infty$ and $a_\infty\in(0,1)$, has been
studied by several authors \cite{GJS,Kon02,Kon05}. It is a particular case of the
isotropic case (b) above. Because $U=SC_\infty$ has no eigenvalue whenever
$a_\infty\in(0,1)$, \eqref{muV_iso} reduces to
$$
\mu_\V(\d\upsilon)= w(\upsilon)f_{\rm K}(\upsilon,a_\infty)\;\!\d\upsilon
$$
with
$$
w(\upsilon)=\tfrac12\chi_{[-a_\infty,0)}(\upsilon)w_\ell(\upsilon)
+\tfrac12\chi_{(0,a_\infty]}(\upsilon)w_{\rm r}(\upsilon).
$$
As noted in Remark \ref{rem_wlt}, the outgoing states determines $w_\star$ uniquely.
However, in the homogeneous case, the initial state $\Psi_{\rm in}$ and the coin matrix
$C_\infty$ determine the function
$
w\in\lone\big([-a_\infty,a_\infty],
\tfrac12f_{\rm K}(\;\!\cdot\;\!,a_\infty)\;\!\d\upsilon\big)
$
uniquely. For a special initial state $\Psi_{\rm in}$ and coin matrix $C_\infty$, $w$
can be computed explicitly. We refer to the works \cite{Ma1,Ma2} of Machida for more
information on the relation between the limit distribution and the initial state.
\end{Example}



\end{document}